\documentclass[%
 reprint,
 amsmath,
 amssymb,
 aps,
 prl%
]{revtex4-2}
\usepackage{amsthm}
\usepackage{cases} 
\usepackage{mathtools} 
\usepackage{physics}
\usepackage{enumitem}
\usepackage{bbm}
\usepackage{tikz}  
\usepackage[hidelinks,breaklinks]{hyperref}
\usepackage{thm-restate}
\usepackage{thmtools,nameref,cleveref}

\crefname{equation}{Eq.}{Eqs.}
\crefname{appendix}{S.M.~Section}{S.M.~Sections}
\Crefname{appendix}{Supplementary Material}{Supplementary Material}
\usepackage{apptools}
\newcommand{\nocontentsline}[3]{}
\newcommand{\tocless}[2]{\bgroup\let\addcontentsline=\nocontentsline#1{#2}\egroup}

\makeatletter
\def\maketitle{
\@author@finish
\title@column\titleblock@produce
\suppressfloats[t]}
\makeatother

\declaretheorem[name=Theorem,
refname={theorem,theorems},
Refname={Theorem,Theorems}]{theorem}
\declaretheorem[name=Proposition,
refname={proposition,propositions},
Refname={Proposition,Propositions}]{prop}
\declaretheorem[name=Lemma,
refname={lemma,lemmas},
Refname={Lemma,Lemmas},
numberlike=prop]{lemma}
\declaretheorem[name=Corollary,
refname={corollary,corollaries},
Refname={Corollary,Corollaries},
numberlike=prop]{corollary}
\declaretheorem[name=Definition,
refname={definition,definitions},
Refname={Definition,Definitions},
numberlike=prop]{definition}
\theoremstyle{remark}
\newtheorem{remark}[prop]{Remark}
\newtheorem*{remark*}{Remark}

\AtAppendix{\numberwithin{prop}{section}}
\AtAppendix{\numberwithin{lemma}{section}}
\AtAppendix{\numberwithin{theorem}{section}}
\AtAppendix{\numberwithin{definition}{section}}
\AtAppendix{\numberwithin{corollary}{section}}

\newcommand{\proj}[1]{\ketbra{#1}{#1}}
\newcommand{\TrX}[2]{\text{Tr}_{#1}\left[ #2 \right]}
\newcommand{\rr}{{\mathbbm{R}}}
\newcommand{\rrp}{{\mathbbm{R}_{>0}}}

\newcommand{\cc}{{\mathbbm{C}}}
\newcommand{\kk}{{\mathbbm{K}}}
\newcommand{\id}{{\mathbbm{1}}}

\newcommand{\Gs}{\Gamma}

\newcommand{\G}[1]{\Gamma\left\{#1\right\}}

\newcommand{\Gi}[2]{\Gamma_{#1}\left\{#2\right\}}

\newcommand{\ReP}[1]{\mathrm{Re}\left( #1 \right)}
\newcommand{\ImP}[1]{\mathrm{Im}\left( #1 \right)}
\newcommand{\I}[1]{\overline{\id}^{{\scriptscriptstyle( #1 )}}}
\newcommand{\J}[1]{\overline{J}^{{\scriptscriptstyle( #1 )}}}
\newcommand{\Hilb}[1]{\mathcal{H}^{#1}} 
\newcommand{\LinOp}[1]{\mathcal{L}\left( #1 \right)} 
\newcommand{\rtensor}{%
	\mathbin{%
		\tikz[baseline=0.1ex]{%
			\draw[rounded corners=2pt] (0,0) rectangle (0.8em,0.8em);
			\draw (0.2em,0.2em) -- (0.6em,0.6em);
			\draw (0.2em,0.6em) -- (0.6em,0.2em);
		}%
	}%
}
\newcommand{\rdottensor}{%
	\mathbin{%
		\tikz[baseline=0.1ex]{%
			\draw[rounded corners=2pt] (0,0) rectangle (0.8em,0.8em);
			\fill (0.4em,0.4em) circle (0.1em);
		}%
	}%
}
\newcommand{\rwedgetensor}{%
	\mathbin{%
		\tikz[baseline=0.1ex]{%
			\draw[rounded corners=2pt] (0,0) rectangle (0.8em,0.8em);
			\draw (0.2em,0.6em) -- (0.4em,0.2em);
			\draw (0.6em,0.6em) -- (0.4em,0.2em);
		}%
	}%
}
\newcommand{\Herm}{\mathsf{Herm}}
\newcommand{\HermC}[1]{\mathsf{Herm}_{#1}\left(\cc\right)}
\newcommand{\HermP}[1]{\mathsf{Herm}^+_{#1}\left(\cc\right)}

\newcommand{\AHermC}[1]{\mathsf{AHerm}_{#1}\left(\cc\right)}

\newcommand{\SymX}[2]{\mathsf{Sym}_{#1}\left(#2\right)}
\newcommand{\SymP}[1]{\mathsf{Sym}^+_{#1}\left(\rr\right)}

\newcommand{\ASymX}[2]{\mathsf{ASym}_{#1}\left(#2\right)}
\newcommand{\EffCom}{\mathcal{E}(\mathcal{H})}
\newcommand{\StatCom}{\mathcal{S}(\mathcal{H})}
\newcommand{\SY}{\mathsf{SY}}
\newcommand{\SYR}[1]{\mathsf{SY}_{#1}(\rr)}
\newcommand{\SYRP}[1]{\mathsf{SY}^+_{#1}(\rr)}
\newcommand{\EffR}{\mathsf{E}_\textup{Re}}
\newcommand{\StatR}{\mathsf{S}_\textup{Re}}
\newcommand{\Ssep}{\mathsf{S}^\textup{SEP}_\textup{Re}}
\newcommand{\Esep}{\mathsf{E}^\textup{SEP}_\textup{Re}}
\newcommand{\SsepC} {\mathcal{S}^\textup{SEP}(\mathcal{H})}  
\newcommand{\EsepC}{\mathcal{E}^\textup{SEP}(\mathcal{H})}

\begin{document}
\setcounter{secnumdepth}{2} 
\title{The Locality of Quantum Theory Does Not Require Complex Numbers}
\author{Timoth{\'e}e Hoffreumon}
 \email{t.hoffreumon@gmail.com}
 \affiliation{
 Quantum Computation Structures (QuaCS) Inria Team, LMF, ENS Paris-Saclay, Université Paris-Saclay, Gif-sur-Yvette, France}
 \affiliation{Mathematical Institute, Slovak Academy of Sciences,  Bratislava, Slovakia}
 \affiliation{Basic Research Community for Physics e.V., Leipzig, Germany}

\author{Mischa P. Woods}
 \email{mischa.woods@gmail.com}
\affiliation{ENS Lyon, Inria, France}
\affiliation{University Grenoble Alpes, Inria, Grenoble, France}
\affiliation{ETH Zurich, Zurich, Switzerland}

\date{\today}

\begin{abstract}
Complex numbers are central to the standard formulation of quantum theory, but whether they are fundamentally necessary remains unsettled. Familiar real formulations reproduce quantum predictions but fail to be representation-local: local physical operations need not correspond to local entities in the mathematical representation. Here we show that representation-local formulations based solely on real numbers nevertheless exist. Our construction is based on the observation that not every structural feature of finite-dimensional density-matrix quantum theory over the complex numbers need be retained in a real formulation. Specifically, we reconsider two features inherited from the complex formalism: the use of the Kronecker product as the system-combination rule and the representation of states and effects by positive semidefinite matrices. Relaxing only one of these assumptions at a time, we first obtain a representation-local real quantum theory with a modified combination rule that closely parallels the conventional complex formulation. We then show that even with the usual Kronecker-product rule a representation-local real formulation exists, provided states and effects are not required to be positive.
\end{abstract}

\maketitle

\section{Introduction}
Quantum theory (QT) required a huge conceptual departure from the classical theories which came before it. It also manifestly used complex numbers. Many have therefore questioned whether these complex numbers are necessary and intrinsically linked to the new phenomena|such as superposition or entanglement|or whether they are merely a mathematical convenience. The question has a long and rich history. A formulation of quantum theory only involving real numbers was first theorised by Jordan, Wigner, and von Neumann as early as 1934~\cite{JvNW1933}, and was subsequently studied in the 60s mainly by St\"{u}ckleberg, but also others (see e.g., Ref.~\cite{Stueckelberg1960,Stueckelberg1961II,Stueckelberg1961III,Threefold,Myrheim1999,McKague,Aleksandrova_2013}). 
In this approach, which we call a \emph{\textbf{Real Representation of Quantum Theory (RRQT)}}, complex quantum theory in $d$ dimensions is replaced by a real representation reproducing the same measurement statistics using states and effects in at least $2d$ dimensions and only real numbers.

However, when a quantum system is multipartite and its subsystems are associated with tensor factors of the Hilbert space, its description in RRQT is not necessarily \textbf{representation-local}: operators associated with a given subsystem may act nontrivially on other tensor factors. For example, in (density-matrix) QT, the independent preparation of two subsystems $A$ and $B$ is described by a product state $\rho_A \otimes \sigma_B$ on $\mathcal{H}_A \otimes \mathcal{H}_B$, where the preparation of subsystem $A$ is fully described by the local operator $\rho_A$ acting on $\mathcal{H}_A$ alone. By contrast, in the RRQTs of Refs.~\cite{Myrheim1999,McKague,Aleksandrova_2013}, encoding complex state amplitudes using real vectors requires the description of the preparation of $A$ to involve operators acting nontrivially either on an additional fiducial subsystem (the ``universal bit'') or on subsystem $B$~\footnote{A more detailed presentation of representation-locality is given in the end matter, Sec. \ref{sec:representation-locality}.}. Thus, although the subsystems are local, their real representation involves manifestly nonlocal elements.

Consequently, RRQT cannot provide a real and representation-local description of experiments. An alternative is the popular toy model known as \emph{\textbf{Real Quantum Theory (RQT)}}, obtained by replacing complex numbers with real numbers throughout the standard Hilbert-space formalism, thereby retaining representation-locality. As noted by Araki~\cite{Araki1980} and later popularized by Wootters~\cite{Wootters1990}, RQT is however not locally tomographic: even with infinitely many copies, the global state of a composite system cannot be reconstructed from local measurements alone. Making local tomography an axiom of physical theories has therefore been a standard way to rule out RQT in favour of QT.

Still, tomographic nonlocality does not directly rule out RQT or any other real and representation-local analogues of quantum theory as consistent descriptions of nature. To move to arguments based on predicted differences in observation, one can consider the violation of Bell inequalities instead. In this context, the question of real numbers arises from the observation that the CHSH inequality violation is described using only real operators in quantum theory. 
Since restricting the operators in QT to only those that are real for a given basis provides a RQT description, it is then legitimate to ask whether every Bell-like inequality violation has one~\cite{Gisin_2007}. While positive answers were known for specific cases~\cite{Pal_2008,Moroder_2013}, Renou et al.~\cite{Renou2021} answered this question negatively by identifying a Bell-type inequality violation that cannot be reproduced in RQT. Their result has been interpreted as evidence that complex numbers are a fundamental ingredient of quantum theory. This would contrast with classical wave theory, where real numbers suffice and complex numbers can be introduced purely for technical convenience. 

Or at least, this is how the story goes. One can argue that some axioms used in such no-go theorems were formulated in their familiar form because they \emph{already} assumed complex numbers. The most prominent example is the system-combination rule: the tensor product, denoted $\otimes$, which is featured explicitly in postulates encoding system locality. Historically, the complex Hilbert-space formalism was developed without the tensor-product structure at first. At that time, independent parts of a system (\textit{subsystems}) were characterized by requiring operators acting on different subsystems to commute.
Only later---for example during the 1925--1930 period of Dirac's career leading up to his book~\cite{DiracBook}---did a combination rule begin to crystallise, eventually leading to the tensor product. One of its first rigorous presentations then appears in von Neumann's book~\cite{NeumannBook}.

In the commuting-subsystems paradigm, St\"{u}ckelberg's construction~\cite{Stueckelberg1960} provides an RRQT which is representation-local by construction: subsystems are represented by commuting subalgebras, and locality is defined through commutativity. By contrast, Renou et al.~\cite{Renou2021} showed that, when subsystems are represented by tensor factors and locality is defined through the Kronecker product, the resulting real theory, RQT, can neither be nor contain an RRQT. 

Indeed, the RRQTs of Refs.~\cite{Myrheim1999,McKague,Aleksandrova_2013} do not provide representation-local descriptions \emph{assuming representation-locality to be defined via the Kronecker product}. However, this is a convention inherited from QT. Over the reals, this representation-nonlocality does not imply Bell-nonlocality: the correlations obtained from a Bell-local quantum state remain Bell-local when calculated using its RRQT representation. This raises the question of whether the no-go result persists if one modifies the state space or the system-combination rule in order for representation-locality to better reflect Bell-locality, or whether it is specific to RQT.

In this Letter, we show the latter by presenting two representation-local real theories that exactly reproduce all predictions of quantum theory, including those for Bell-nonlocal experiments. 
The first can be understood as showing that the previous nonlocal approaches~\cite{Myrheim1999,McKague,Aleksandrova_2013} can be made representation-local by changing the system-combination rule. 
The second can be understood as the RRQT obtained from RQT by redefining its sets of states and effects.

\section{Axioms}\label{sec:Axioms}
We call {\bf\emph{Complex Quantum Theory}} (CQT) the standard textbook formulation of quantum theory. Following~\cite{Renou2021} (who cite~\cite{axioms1,axioms2}), CQT is defined from four postulates (here stated at the level of density matrices): (1) For every physical system $\mathsf{S}$, there corresponds a real vector space of Hermitian matrices $\Herm_\mathsf{S}$ and its state is represented by a positive semi-definite normalized vector $\Phi\geq 0$ in $\Herm_\mathsf{S}$, that is, $\Tr[\Phi]=1$. (2) A measurement $\Pi$ in $\mathsf{S}$ corresponds to an ensemble $\{\Pi_r \geq 0\}_r$ of so-called effects in $\Herm_\mathsf{S}$, and indexed by the measurement result $r$, with $\sum_r \Pi_r = \id_\mathsf{S}$. (3) Born rule: if we measure $\Pi$ when system $\mathsf{S}$ is in state $\Phi$, the probability of obtaining result $r$ is given by $\textup{Pr}(r)= \Tr[\Pi_r \Phi]$. (4) The Hermitian vector space $\Herm_\mathsf{ST}$, corresponding to the composition of two systems $\mathsf{S}$ and $\mathsf{T}$ is $\Herm_\mathsf{S}\otimes \Herm_\mathsf{T}$. The operators used to describe measurements or transformations in system $\mathsf{S}$ act trivially on $\mathsf{T}$ and vice versa. Similarly, the state representing two independent preparations of the two systems is the tensor product of the two preparations.

A \textbf{real theory} is an alternative to CQT for which the underlying Hilbert space is a real vector space rather than a complex one, i.e. one where $\Herm_\mathsf{S}$ is a space of real symmetric matrices instead of complex Hermitian. RQT is the real theory tied to the exact formulation of the above postulates. But when considering the existence of real alternatives, it is important to question which aspects of these axioms are fundamental to CQT itself, and which are merely a consequence of assuming complex numbers from the outset---as is the case in all conventional axiomatic formulations of CQT such as the one above. 

As mentioned above, the tensor product in postulate (4) entered only after the single-system formalism had been established, so it should not be treated as fundamental. What matters are the locality properties encoded by postulate (4), not its standard finite-dimensional realization using the Kronecker product $\otimes$. We therefore relax (4) by replacing $\otimes$ with a to-be-specified combination rule $\rtensor$ that preserves those locality properties; we call the resulting postulate (4$^\prime$).

The density-matrix formulation also suggests another relaxation: over the reals, positivity of states and effects in (1) and (2) may be weakened provided postulate (3) still yields well-defined probabilities. We denote the resulting postulates by (1$^\prime$) and (2$^\prime$).

We now show that either relaxation leads to a real theory that locally represents quantum theory. We will call any RRQT obeying (1$^\prime$), (2$^\prime$), (3), and (4$^\prime$) a \textbf{Real-Number Quantum Theory (RNQT)}~\footnote{One could also consider replacing the Hilbert-Schmidt inner product $\langle{\cdot, \cdot}\rangle =\Tr[(\cdot)^\dag (\cdot)]$ appearing in postulate (3) with another bilinear form, for there is no reason that all RRQTs use the same inner product. Indeed, compared to complex quantum theory where (1) and (3) can be shown necessary from (2) and (4) by a Gleason-Wallach-Busch style of argument~\protect\cite{Gleason_1957,Wallach_2000,Busch_2003,Caves_2004}, there are no a priori on (3) other than to be a pairing from states and effects to probabilities. This said, it suffices to use the same inner product for our purposes, and by doing so, we keep our RRQT as close as possible in structure to its complex counterpart.}\nocite{Gleason_1957,Wallach_2000,Busch_2003,Caves_2004}.

\section{Structure}\label{sec:Structure}
The aim is to form a real representation of CQT which is representation-local w.r.t. $n$ subsystems. We require a one-to-one correspondence between CQT states/effects and their real representatives. It follows that there should exist a linear map between these objects. Moreover, said maps should obey the locality properties of the theory in the sense laid out in~\Cref{sec:Axioms}. With this in mind, we consider maps that preserve the product structure of CQT. Since in CQT states and effects live in a real vector space, namely the tensor product space of complex Hermitian matrices, $\Herm_{n,d}(\cc):=\Herm_{d_1}(\cc)\otimes \Herm_{d_2}(\cc)\otimes \ldots \otimes \Herm_{d_n}(\cc)$ (with $\Herm_{d_i}(\cc)$ denoting a $d_i \times d_i$ complex Hermitian matrix), we define the map by how it acts on Kronecker-product elements in $\Herm_{n,d}(\cc)$. 
Explicitly, the maps we use for states and effects, denoted $M(\cdot)$ and $E(\cdot)$ respectively, are of the form
\begin{equation}\label{eq:map M} 
    \begin{aligned}
    	M(\cdot)\!:= a \Gamma_{d_1}\!(\cdot) \rtensor  \Gamma_{d_2}\!(\cdot) \rtensor \ldots \rtensor \Gamma_{d_n}\!(\cdot), \quad  
        E(\cdot)\!:= M(\cdot)/a, 
    \end{aligned}
\end{equation}
where $a\in\rrp$ is a normalization constant such that $M(\cdot)$ preserves the trace, while the map $\Gamma_{d_l}(\cdot)$ acts on the local Hermitian matrices in $\Herm_{d_l}(\cc)$ with output with output in a real matrix space.

We call $\rtensor$ a {\bf{\emph{combination rule}}}. It maps two inputs from two copies of said real vector space to another real vector space of matrices with real coefficients. We assume it to be associative, linear in both inputs, and satisfies the mixed-product property: $(A\rtensor B)\cdot (C\rtensor D)=(A\cdot C)\rtensor (B\cdot D)$, where $\cdot$ is the matrix product. This is necessary in order to define a notion of locality w.r.t. the combination rule; otherwise $(A\rtensor B) \neq (A \rtensor \id) \cdot (\id \rtensor B)$ so the action of the operator $A$ cannot be `localised' as $(A \rtensor \id)$.

The required factorization with respect to $\rtensor$ of $M$ in~\cref{eq:map M} guarantees the RRQT to be representation-local: any pure $\otimes$-product in CQT will be mapped to a pure $\rtensor$-product.  
When such a multilinear combination rule happens to represent the tensor product for the relevant space of Hermitian matrices~\footnote{As a reminder, the definition of the tensor product of both vector and inner-product spaces can be found in the Supplementary Material, Section~\ref{sec:tensor product}. The Supplementary Material includes Refs.~\cite{Jacobson,Nielsen2012}.}, i.e. when $M$ maps unentangled (w.r.t. $\otimes$) basis elements to unentangled (w.r.t. $\rtensor$) basis elements, postulate (4$^\prime$) is satisfied. In the special case where $\rtensor=\otimes$, postulate (4) is satisfied. 

Write each $H\in\Herm_{d_l}(\cc)$ uniquely as $H=H^\textup{Sy}+i H^\textup{An}$, with $H^\textup{Sy}$ real symmetric and $H^\textup{An}$ real antisymmetric. 
The $l$-th local map is $\Gs_{d_l} : \Herm_{d_l}(\cc) \rightarrow \SY_{d_l}(\rr)$, acts:
\begin{equation}\label{eq:local mapping}
    \begin{aligned}
     \Gamma_{d_l}(H) = 
	\left({\begin{array}{cc} H^\textup{Sy} & -H^\textup{An} \\ H^\textup{An} & H^\textup{Sy} \end{array}}\right)= \id\otimes H^\textup{Sy} + J\otimes H^\textup{An},
    \end{aligned}
\end{equation}
where $\id$ and $J$ are real matrix representations of $1$ and $i$ respectively, namely 
$\id:=\begin{pmatrix} 1 & 0 \\ 0 & 1 \end{pmatrix}$,  
$J:= \begin{pmatrix} 0 & -1 \\ 1 & 0 \end{pmatrix}$. 
We denote by $\SY_{d_l}(\rr)$ the real vector space of all such real matrices. We call it `Special Symmetric' since it is a subspace of the space of real symmetric matrices. 
It is a real-linear bijection which preserves the inner product and positive semi-definiteness between $\Herm_{d_l}(\cc)$ and $\SY_{d_l}(\rr)$~\footnote{See \Cref{app:construction} for a proof: the restatement of this claim is given by~\Cref{thm:posvitive semi definitness conservation of the local mapping} and a proof follows.}. 

From the decomposition $H=H^\textup{Sy}+i H^\textup{An}$, one also observes that the CQT states and effects in $\Herm_{2d_l}(\cc)$ which happen to have only real coefficients live in the space of real symmetric matrices, and a subset of these span $\SY_{d_l}(\rr)$. Thus, while we have chosen the mapping~\cref{eq:local mapping} to keep the information about the imaginary part (which is necessary if we want the map to provide a one-to-one correspondence between our real theory and CQT), one can always view the output as a real subspace spanned by the real-only states and effects of a larger-dimensional CQT. 

With the local space defined, the codomain of $M(\cdot)$ is denoted $\SY_{n,d}(\rr):=\SY_{d_1}(\rr)\rtensor\SY_{d_2}(\rr)\rtensor\ldots\rtensor\SY_{d_n}(\rr)$ and is defined as the space of all real linear combinations of states of the corresponding subsystems. 
Since the pair of maps defined by~\cref{eq:map M,eq:local mapping} are linear and invertible~\footnote{See~\Cref{sec:mapping} for a proof.}, the structure of the obtained real theory is already quite rich.

From~\cref{eq:map M}, the sets of states and effects are readily definable. In CQT, we denote the set of quantum states by $\StatCom\subset \Herm_{n,d}(\cc)$, while in our real theory by $\StatR\subset\SY_{n,d}(\rr)$. Similarly, we denote the set of quantum effects by $\EffCom\subset \Herm_{n,d}(\cc)$, while in our real theory by $\EffR\subset\SY_{n,d}(\rr)$. Membership to these sets in our real theory can be defined via the membership condition of their counterparts in CQT: $\rho\in\SY_{n,d}(\rr)$ is in $\StatR$ iff $M^{-1}(\rho)\geq 0$ and $\Tr[M^{-1}(\rho)]=1$, while $e\in\SY_{n,d}(\rr)$ is in $\EffR$ iff $E^{-1}(e)\geq 0$. Since these sets of states and effects in our real theory are in one-to-one correspondence with CQT, postulates (1$^\prime$) and (2$^\prime$) are satisfied by our real theories. In the special case that said sets have only positive semi-definite elements, postulates (1) and (2) are satisfied too. The sets of real states and effects, $\StatR$ and $\EffR$, are also convex because sets $\StatCom$ and $\EffCom$ are convex and the map $M^{-1}(\cdot)$ is linear.

We have already seen how product states and effects in real theories are in one-to-one correspondence with tensor-product states and effects in CQT. Similarly, we define the set of real separable states, denoted $\Ssep$, in analogy with separable states in CQT as follows: $\rho$ is in $\Ssep\subset \SY_{n,d}(\rr)$ iff there exist states $\rho_j^l$ such that $\Gamma^{-1}_{d_l}(\rho_j^l)\geq 0$, $\Tr[\Gamma^{-1}_{d_l}(\rho_j^l)/a]=1$, and a probability distribution $\{p_j\}_{j}$, for which 
$\rho= \sum_j p_j \,\rho_j^1\rtensor \rho_j^2\rtensor\ldots\rtensor \rho_j^n$. We define $\Esep$, the set of real separable effects, in the same way.

The sets $\Ssep$ and $\Esep$ are in one-to-one correspondence with their CQT counterparts. That is, $\rho\in\Ssep$ and $e\in\Esep$ iff $M^{-1}(\rho)\in\SsepC$ and $E^{-1}(e)\in\EsepC$, where $\SsepC$ and $\EsepC$ denote the usual sets of separable states and effects in CQT~\footnote{See~\Cref{app:properties of separable states} for a proof.}. Accordingly, the sets of entangling effects and entangled states are defined as the elements in $\EffR$ and $\StatR$, which are not in $\Esep$ and $\Ssep$, respectively. 

\section{Real-Number Quantum Theories (RNQT\texorpdfstring{\lowercase{s}}{s})}\label{Real quantum theory}
\begin{definition}
We say a particular combination rule $\rtensor$ gives rise to a {\bf{real-number quantum theory}} (RNQT) if it reproduces the same measurement statistics as CQT, namely if
\begin{align}
	 \Tr[e \rho]=\Tr[E^{-1}(e) M^{-1}(\rho)],\label{eq:inner product}
\end{align}
for all effects $e\in\EffR$ and density matrices $\rho\in\StatR$. 
\end{definition}
 Note that since the maps $E$ and $M$ are invertible, it follows that if~\cref{eq:inner product} holds, then it also holds for all complex effects and for all complex states, i.e.~\cref{eq:inner product} holds iff $\Tr[E(e') M(\rho')]=\Tr[e' \rho']$ holds for all $e', \rho'$ positive semi-definite in $\Herm_{n,d}(\cc)$.  

\Cref{eq:inner product} realizes the Born rule [postulate (3)]. It thus follows that RNQTs satisfy postulates (1$^\prime$), (2$^\prime$), (3), and (4$^\prime$).

Under these relaxed postulates, if a real theory is an RNQT, then the set of statistics generated by all measurements on separable states is identical to the set of measurement statistics generated by all measurements on the set of separable states in CQT. Likewise, the same relationship holds for the set of measurement statistics generated by all separable effects~\footnote{See \Cref{app:properties of separable states} for the proof of these two claims.}.

In the special case of only one subsystem, namely $n=1$, $M(\cdot)=\frac{1}{2} \Gamma_{d_1}(\cdot)$, postulate (4) becomes unnecessary. In this case, $v\in\SY_{d_1}(\rr)$ is positive semi-definite iff $M^{-1}(v)$ is positive semi-definite, so the only possible RNQT obeys postulates (1), (2), and (3): it is St\"uckelberg's construction~\footnote{See \Cref{sec:4bis_single_system} for a proof.}. While his construction can be extended to the multipartite case in the commuting-subsystem paradigm, it was never considered in the tensor-factors paradigm. This highlights how the difficulty with devising multipartite real theories is about defining appropriate combination rules.

In light of that, we present two ways of extending this single-system RNQT to a multipartite RNQT.  Whether one of our constructions is the tensor-factors equivalent of St\"uckelberg's is a question left open for future work.

\section{RNQT with a new combination rule}\label{A new combination rule}

The first approach to building a multipartite RNQT relies on relaxing the model-dependent part of postulate (4) by demanding the combination rule $\rtensor$ not to be the Kronecker product $\rtensor \neq \otimes$ but rather another map $\rdottensor$ which verifies our definition of a combination rule. We require:
\begin{definition}
    Let $S$ and $T$ be special symmetric matrices in $\SYR{d_1}$ and $\SYR{d_2}$ such that $S=\begin{pmatrix} S^\textup{Sy} & -S^\textup{An} \\ S^\textup{An} & S^\textup{Sy} \end{pmatrix}$ and $T=\begin{pmatrix} T^\textup{Sy} & -T^\textup{An} \\ T^\textup{An} & T^\textup{Sy} \end{pmatrix}$ like above. Then, we define their combination rule $\rdottensor : \SYR{d_1} \times \SYR{d_2} \rightarrow \SYR{d_1d_2}$ as
    \begin{equation}\label{eq:dottensor}
    	\begin{aligned}
    		S \rdottensor T :=& \id \otimes (S^\textup{Sy} \otimes T^\textup{Sy} - S^\textup{An} \otimes T^\textup{An})\\&
    		+ J \otimes (S^\textup{Sy} \otimes T^\textup{An} + S^\textup{An} \otimes T^\textup{Sy}) \:.
    	\end{aligned}
    \end{equation}
\end{definition}

\noindent This specific definition for the combination rule has been chosen because it commutes with applying the $\Gs$ map,
\begin{equation}\label{eq:factors}
    \Gamma_{d_1 d_2}(H \otimes L) = \Gamma_{d_1}(H) \rdottensor \Gamma_{d_2}(L) \:,
\end{equation}
for $H$ and $L$ in $\HermC{d_1}$ and $\HermC{d_2}$, respectively~\footnote{See~\Cref{sec:4bis_multipartite} for a representation-theoretic explanation for this choice, and the proof of the lemma below, which appears as~\cref{prop:propertiesofdottensor}.}.
From the form this composition has, it is clear that it mimics the behavior of a complex combination of Hermitian matrices under the substitution $\{1,i\}$ $\rightarrow$ $\{\id, J\}$. 

The right-hand side of \cref{eq:dottensor} has appeared in RRQT before~\cite{McKague,Aleksandrova_2013}. However, while this map is nonlocal with respect to the Kronecker product $\otimes$, it is not with respect to $\rdottensor$. Redefining locality with respect to $\rdottensor$ proves the RRQTs~\cite{Myrheim1999,McKague,Aleksandrova_2013} to be representation-local under the proper notion of combination rule. The extra ``resources'' needed to represent the complex phases are no longer interpreted as physical; instead, they are now part of the definition of $\rdottensor$. 

This redefinition makes sense as the $\rdottensor$ map is precisely the tensor product for spaces of special symmetric matrices, as the following lemma shows:
\begin{lemma}\label{lem:properties of rdottensor}
    The combination $\rdottensor$ is an associative real-bilinear map that obeys the mixed-product property such that $\dim{\SYR{d_1} \rdottensor \SYR{d_2}} = \dim{\SYR{d_1}} \dim{\SYR{d_2}}$, and $\SYR{d_1} \rdottensor \SYR{d_2} = \SYR{d_1d_2}$.
\end{lemma}
Hence, $\rdottensor$ is a notion of combination rule which can be used in accordance to postulate (4). Contrastingly, the usual Kronecker product $\otimes$ cannot fulfill this role, as it does not even map to the proper size of matrices ($\SYR{d_1d_2} \subset \rr^{2d_1 d_2\times 2 d_1 d_2}$ whereas the Kronecker  product maps to $\rr^{4 d_1 d_2 \times 4 d_1 d_2}$). 

With this last piece of the puzzle, we can show that $\rdottensor$ yields a RNQT through the map $\Gs_{d_l}(\cdot)$ since it follows from~\cref{eq:factors} that $M(\cdot)$ in~\cref{eq:map M} simply takes on the form $M(\cdot)=\frac{1}{2} \Gamma_{d_1 d_2 \ldots d_n}(\cdot)$. With this, the previously-discussed single-system properties of the map $\Gamma_{d_l}(\cdot)$ can be exported to the entire multi-partite system: since the local mappings $\Gs_{d_i}(\cdot)$ are positivity-preserving bijections between CQT and RNQT, the following theorem follows
\begin{theorem}\label{thm:new combination rule}
	The choice of combination rule $\rtensor = \rdottensor$,  gives rise to a RNQT with positive semi-definite states and effects, namely $e$ is in $\EffR$ and $\rho$ is in $\StatR$ iff $e\geq 0$ and $\rho \geq 0,$ $\Tr[\rho]=1$. 
\end{theorem}

Thus, real representations satisfying postulates (1), (2), (3), and (4$^\prime$) can match the predictive power of CQT.

Since $\rdottensor$ is the image of $\otimes$ from CQT, an RNQT also exists inside RQT once (4) is relaxed: one simply restricts states and effects from all symmetric matrices to the special symmetric ones. Hence, imposing the Kronecker product as the combination rule is what singles out RQT.

One can likewise define a real, representation-local pure-state theory by taking the image of the rank-1 projectors of CQT.

Even when an RNQT does not use $\rtensor=\rdottensor$, $\rdottensor$ remains useful because it lets one determine whether $v\in\SY_{n,d}(\rr)$ is a state or effect without explicitly evaluating $M^{-1}(v)$. 
With this in mind, for any RNQT with combination rule $\rtensor$, we define the map $G(v_{d_1} \rtensor v_{d_2} \rtensor\ldots \rtensor v_{d_n}):= v_{d_1} \rdottensor v_{d_2} \rdottensor\ldots \rdottensor v_{d_n}$, where $v_{d_l}\in\SY_{d_l}(\rr)$. It extends to all vectors in $\SY_{n,d}(\rr)$ by linearity.
\begin{theorem}[\footnote{See~\Cref{Sec:proof of state set determination without complext numbers} for the proof.}] \label{thm:no need for complex numbers}
For any RRQT with combination rule $\rtensor$, consider $v\in\SY_{n,d}(\rr)$. It holds that $M^{-1}(v)\geq 0$ iff $G(v)\geq 0$.
\end{theorem}

The criterion $G(v)\geq 0$ involves only real numbers; in a local product basis it merely swaps basis elements locally. The only nontrivial step is then checking positive semi-definiteness, exactly as in CQT.

\section{RNQT with Kronecker-product combination rule}\label{sec:construction_Kronecker}
We now consider the case in which one uses the standard combination rule in CQT, namely $\rtensor=\otimes$.

\begin{theorem}[\footnote{See~\Cref{sec:proof thm normal tensor product} for the proof.}] \label{thm:normal tensor product}
The choice of combination rule $\rtensor = \otimes$ gives rise to a RNQT. 
Moreover, if a state is separable, then it is positive semi-definite, namely if $\rho\in\Ssep$, then $\rho \geq 0$. Similarly, the same holds for effects: if $e\in\Esep$, then $e \geq 0$.
Furthermore, there exist entangled effects and entangled pure-states in $\EffR$ and $\StatR$ respectively, which are not positive semi-definite.
\end{theorem}
This RNQT thus satisfies postulates (1$^\prime$), (2$^\prime$), (3) and (4). One direct consequence of this theorem is that the non-positive-semi-definiteness of a state or effect is an entanglement witness. 

Unlike with the choice of combination rule $\rtensor=\rdottensor$ from~\Cref{A new combination rule}, for the RNQT arising from the combination rule $\rtensor=\otimes$, one cannot define a real-number  theory at the level of kets representing states. This is because~\Cref{thm:normal tensor product} tells us that at least one of the states will have at least one negative eigenvalue. 
This said, they are still convex and characterisable via~\Cref{thm:no need for complex numbers}.

Nonetheless, this second RNQT differs from all the theories considered so far in being explicitly partition-dependent. What qualifies as a valid state or effect depends on the chosen factorisation into subsystems. For example, interpreting a system as a ququarter or as a pair of qubits yields two distinct sets of valid states, each of different dimension.

\section{Conclusions and outlook}\label{sec:Conclusions and outlook}
We showed that a quantum theory based only on real numbers is possible once model-dependent assumptions in the standard postulates are replaced by their structural content.

In particular, the usual Kronecker-product rule for composing multipartite systems need not be treated as fundamental. Relaxing that assumption yields a tomographically local RNQT whose states and effects satisfy the usual positivity requirement. Alternatively, one may retain the conventional Kronecker product and instead relax positivity of density matrices; this yields a natural RNQT that still respects representation-locality. In that theory, operator negativity appears only for entangled states or entangling effects, and therefore serves as an entanglement witness.

The resulting real representations have roughly twice the dimension of their complex counterparts, as expected, but otherwise introduce only minor inconveniences relative to CQT. This supports the view that complex numbers function mainly as a mathematical convenience rather than an experimentally mandated feature. Our constructions may also be useful in practice by making some structures more explicit; for instance, antiunitary transformations of CQT become linear maps in RNQT~\footnote{See~\Cref{sec:4bis_single_system} for a worked out example in the case of the RNQT from \cref{A new combination rule}.}.

\begin{acknowledgments}
{\bf Note Added}: 
While this work was in preparation, the authors became aware of related results~\cite{Barrios2026}. Subsequently, Refs.~\cite{Feng2025,Volovich2025,Arefeva2025} appeared, addressing related aspects of real quantum theory. In particular, Ref.~\cite{Feng2025} seemingly arrives at the opposite conclusion to this Letter. However, this is not a contradiction once it is recognized that its notion of locality corresponds to ``representation-local with respect to the Kronecker product'' in our nomenclature. Furthermore, Refs.~\cite{Cannot,Maioli2026}, which appeared while this work was under review, raise issues with the purported falsification experiment of Ref.~\cite{Renou2021}

{\bf Acknowledgments}: 
We acknowledge extensive and useful comments from an anonymous referee, which greatly helped reshape this work into the present letter.

We thank useful conversations with Ognyan Oreshkov, Atul Arora, Lionel Dmello, Marco Erba, Thomas Galley, Rob Spekkens and Christopher Chubb. We would also like to thank Ravi Kunjwal for putting the authors in contact with one another while at a conference at Perimeter Institute.  

T.H. received support from the ID \#62312 grant from the John Templeton Foundation, as part of ‘The Quantum Information Structure of Spacetime’ Project (QISS). The opinions expressed in this publication are those of the author(s) and do not necessarily reflect the views of the John Templeton Foundation. T.H. also acknowledges support from the Slovak grants APVV-22-0570 and VEGA 2/0128/24. 
M.W. acknowledges support from the Swiss National Science Foundation (SNSF) via an AMBIZIONE Fellowship (PZ00P2\_179914).
\end{acknowledgments}

\bibliography{bibliography}

\setcounter{section}{8}
\section*{End matter}

\subsection{Representation-locality}\label{sec:representation-locality}
The notion of representation-locality used in this Letter refers to the compatibility between the physical identification of subsystems and their mathematical representation. A representation is \emph{representation-local} if local preparations, transformations, and measurements are represented by mathematical objects associated with the corresponding physical subsystem, without requiring nontrivial action on other subsystems. This notion is distinct from Bell locality: it concerns how physical a priori are mirrored by the mathematics, rather than by the correlations they produce.

More precisely, let $A$ and $B$ denote two subsystems. Representation-locality requires the mathematical formalism to distinguish operations on $A$ from operations on $B$ in a manner consistent with their physical identification. In particular, a transformation acting on $A$ alone must be represented by an object that acts trivially on $B$, and conversely. The precise mathematical condition expressing this requirement depends on how the subsystems are represented in the chosen formalism.

In the \emph{commuting-subsystems paradigm}, subsystems are represented by commuting subalgebras $\mathcal{A}_A$ and $\mathcal{A}_B$ of the global algebra of observables. Local observables of $A$ belong to $\mathcal{A}_A$, while those of $B$ belong to $\mathcal{A}_B$, with $[\mathcal{A}_A,\mathcal{A}_B]=0$.

Here, locality is expressed through commutativity: mathematical objects associated with distinct subsystems commute, and operations on one subsystem are represented within its corresponding subalgebra. This implies in particular that the unitary or channel (CPTP-map) evolutions of distinct subsystems commute, and hence so do the density matrices encoding their locally prepared states.

In the \emph{tensor-factor paradigm}, the composite Hilbert space is instead specified by a factorisation $\mathcal{H}=\mathcal{H}_A\otimes\mathcal{H}_B$. Local observables and transformations are represented using the Kronecker-product embedding. For example, an observable $O_A$ acting on $A$ alone is represented globally by $ O_A\otimes I_B$, a local unitary $U_A$ by $U_A\otimes \id_B$, and a local CPTP map $\mathcal{M}_A$ by $\mathcal{M}_A \otimes \mathcal{I}_B$, with $ \mathcal{I}_B$ the identity map on $B$. Thus, locality requires the global representative of an operation on $A$ to act trivially on the factor associated with $B$, and locally prepared subsystems are guaranteed to be jointly described by a product state. This is the standard, Kronecker-product notion of representation-locality considered in this Letter.

These two notions of locality are compatible with their respective subsystem paradigms, but they need not be equivalent when applied to a particular mathematical representation. In particular, a real representation may preserve the commutation relations between subsystem algebras while failing to represent every local operation by an operator of the form $O_A \otimes I_B$ with respect to a chosen tensor-factor decomposition. Such a representation is local in the commuting sense but not necessarily in the Kronecker-product sense.

Finally, as suggested by its name, representation-locality is a property of a specified mathematical representation, not of every mathematically equivalent formulations that use different subsystem identifications. It should therefore be distinguished both from the physical locality of the operations being represented and from Bell locality, which constrains observable correlations. A representation that is nonlocal in the Kronecker-product sense may nevertheless reproduce the Bell-local correlations of a given quantum state. Conversely, reproducing Bell-nonlocal correlations does not by itself imply that the mathematical representation is nonlocal. The distinction is essential when comparing real formulations of quantum theory that adopt different rules for representing subsystems and combining their mathematical descriptions.

The distinction can be illustrated by the RRQT of McKague et al.~\cite{McKague}, in which each complex subsystem $X$ is encoded using a real subsystem $X$ together with an additional real qubit $x$. A local operation on $A$ is then represented by an operator acting on an enlarged subsystem $Aa$ associated with Hilbert space $\mathcal{H}_A \otimes \mathcal{H}_a$, while acting trivially on $Bb$. This operation is local relative to the enlarged partition $Aa\mid Bb$. Nonetheless, representation-local unitaries do not guarantee representation-locality of every other element in the formalism. In this case, the real representation of the joint state of $AaBb$ must always be entangled across $a \mid b$. Therefore, such a state fails to be representation-local in the sense of having been obtained through a pair of local preparation maps (i.e., entanglement-breaking maps) of the form $\mathcal{M}_{Aa} \otimes \mathcal{N}_{Bb}$ since these would produce a non-entangled state w.r.t. the partition $Aa\mid Bb$. 

\subsection{Dynamics: an additional postulate}\label{sec:additional ostulate}
The postulates presented (1)-(4) do not concern dynamics. While common for postulates of quantum theory, those of quantum mechanics requires a fifth postulate, which for CQT is often formulated as follows. 
(5) The evolution of an initial state $\rho$ in a closed quantum system $\mathsf{S}$ is described by the von Neumann equation, (also known as the density-matrix form of the Schr\"odinger equation or Liouville–von Neumann equation), namely
\begin{align}
	\hbar \frac{d \rho}{dt}= -i [H,\rho].
\end{align}
Here $H$ is a Hermitian operator acting on the Hilbert space, $\mathcal{H}_\mathsf{S}$, called the Hamiltonian, and $\hbar \in\rr$ is the reduced Planck constant.
As is well known, this gives rise to unitary dynamics, i.e., $\rho(t)=U(t) \rho U^\dag(t)$.  We discuss in the Supplementary Material (Proposition B.13) how unitary transformations behave under the map $\Gamma_d(\cdot)$. In short, one has $\Gamma_d(\rho(t))= \Gamma_d(U(t)) \Gamma_d(\rho) \Gamma_d(U(t))^T$, where $\Gamma_d(U(t))$ is orthogonal, i.e., $\Gamma_d(U(t))\Gamma_d(U(t))^T= \Gamma_d(U(t))^T \Gamma_d(U(t))= I$ but also symplectic. 
Observe that in  $d$-dimensions $i[H, \rho]=[i H,\rho]$, where $i H$ is an anti-Hermitian matrix. With this observation, the RNQT counterpart of postulate (5) can be formulated. The real formulation consists in replacing the complex operators $\rho$ and $i H$ with their real counterparts. In this case, Hermitian matrices are replaced with symmetric matrices while anti-Hermitian matrices are replaced with anti-symmetric matrices ($M^T=-M$). We denote this postulate (5$^\prime$).We focus on the RNQT of~\Cref{thm:new combination rule} which recall satisfies postulates (1), (2), (3), and (4$^\prime$).   By directly applying our transformation $\Gamma_d$ to both sides of the von Neumann equation, we find a real matrix version which reproduces the dynamics and is consistent with  postulates (1), (2), (3), and (4$^\prime$). Namely
\begin{align}
	\hbar \frac{d \Gamma_d(\rho)}{dt} =- (J  \otimes \id_\mathsf{S})\big[  \Gamma_d (H), \Gamma_d(\rho)\big].
\end{align}
One can readily verify that $(J  \otimes \id_\mathsf{S})\big[  \Gamma_d (H), \Gamma_d(\rho)\big]= \big[ (J  \otimes \id_\mathsf{S}) \Gamma_d (H), \Gamma_d(\rho)\big]$ where  $(J  \otimes \id_\mathsf{S}) \Gamma_d (H)$ is anti-symmetric as required (in analogy to how $i H$ is anti-symmetric).

\clearpage
\appendix
\crefalias{section}{appendix}
\crefalias{subsection}{appendix}

\title{Supplementary Material for \textit{The Locality of Quantum Theory Does Not Require Complex Numbers}}
\maketitle

\onecolumngrid

\section*{Contents}
\begin{enumerate}[label=\textbf{\Alph*}.]
	\item \hyperref[sec:tensor product]{The tensor and Kronecker products} 
	\begin{enumerate}[label=\textbf{\arabic*}.]
		\item \hyperref[sec:tensor product definition]{Definition of the tensor product}
		\item \hyperref[sec:Matrix representations of the tensor-product]{Matrix representation of the tensor product}
	\end{enumerate}
    \item \hyperref[app:construction]{Construction of the mapping from CQT to RRQT}
    \begin{enumerate}[label=\textbf{\arabic*}.]
        \item \hyperref[sec:4bis_single_system]{Positivity-preserving mapping on a single system}
        \item \hyperref[sec:4bis_multipartite]{Positivity-preserving mapping on two subsystems}
    \end{enumerate}
    \item \hyperref[sec:mapping]{Explicit construction of the $M(\cdot)$ map}
	\item \hyperref[app:properties of separable states]{Properties of real separable states}
    \item \hyperref[Sec:proof of state set determination without complext numbers]{Proof of~\Cref{thm:no need for complex numbers}}
    \item \hyperref[sec:proof thm normal tensor product]{Proof of~\Cref{thm:normal tensor product}}
\end{enumerate}

\section{The tensor and Kronecker products}\label{sec:tensor product} 

In this section, we state the usual definitions of both the tensor product and the Kronecker product. We do so for completeness and clarity since these are important to this work.

\subsection{Definition of the tensor product}\label{sec:tensor product definition}

A generic $d$-dimensional vector space over the field $\kk$ is denoted $\mathrm{V}_d(\kk)$, and the fields used in this manuscript are either $\kk=\cc$ or $\kk=\rr$. We will specify the vector space later.  Given two vector spaces, $\mathrm{V}_m(\kk)$ and $\mathrm{V}_n(\kk)$, the tensor product is a method for constructing a new vector space denoted $\mathrm{V}_m(\kk)\otimes \mathrm{V}_n(\kk)$. The elements of $\mathrm{V}_m(\kk)\otimes \mathrm{V}_n(\kk)$ are linear combinations of `tensor products' $v\otimes w$ of elements $v\in \mathrm{V}_m(\kk)$, and $w\in \mathrm{V}_n(\kk)$.  In particular, if $\{v_i\}_i$, and $\{w_j\}_j$ are orthonormal bases for spaces $\mathrm{V}_m(\kk)$ and $\mathrm{V}_m(\kk)$, then $\{ v_i \otimes w_j\}_{i,j}$ is a basis for $\mathrm{V}_m(\kk)\otimes \mathrm{V}_n(\kk)$.

With this in mind, the tensor product of vector spaces is defined through the following properties:
\begin{definition}[Tensor product of vector spaces]\label{def:tensor product}
	We say that a map $\otimes: \mathrm{V}_m(\kk) \times \mathrm{V}_n(\kk) \rightarrow \mathrm{V}_{mn}(\kk)$ is a balanced product if it satisfies the following three properties:
	\begin{itemize}
		\begin{subequations}
			\item [(1)] For an arbitrary scalar $z\in\kk$ and elements $v$ of $\mathrm{V}_m(\kk)$ and $w$ of $\mathrm{V}_n(\kk)$,
			\begin{equation}
				z( v \otimes w)=(z v) \otimes w= v \otimes(zw) \:.
			\end{equation}
			\item[(2)] For arbitrary $ v_1$ and $ v_2$ in $\mathrm{V}_m(\kk)$ and $ w $ in $\mathrm{V}_n(\kk)$,
			\begin{equation}
				\left( v_1+ v_2\right) \otimes w = v_1 \otimes w + v_2 \otimes w \:.
			\end{equation}
			\item[(3)] For arbitrary $ v$ in $\mathrm{V}_m(\kk)$ and $ w_1 $ and $ w_2 $ in $\mathrm{V}_n(\kk)$,
			\begin{equation}
				v \otimes\left( w_1 + w_2 \right)= v \otimes w_1 + v \otimes w_1 \:.
			\end{equation}
		\end{subequations}
	\end{itemize}
	The \textbf{\emph{tensor product}} of $\mathrm{V}_m(\kk)$ with $ \mathrm{V}_n(\kk)$, denoted $\mathrm{V}_m(\kk) \otimes \mathrm{V}_n(\kk)$, is the set spanned by the balanced products between every elements of $\mathrm{V}_m(\kk)$ and $ \mathrm{V}_n(\kk)$ and such that the balanced product can be uniquely identified with a map $\mathrm{V}_m(\kk) \otimes \mathrm{V}_n(\kk) \rightarrow \mathrm{V}_{mn}(\kk)$.
\end{definition}
This is the standard text-book definition, here adapted from~\cite{Jacobson} p.126 for example. (Note that some authors, like e.g., \cite{Nielsen2012}, demand an additional property, namely that $\mathrm{V}_m(\kk)\otimes \mathrm{V}_n(\kk)$ is an $mn$ dimensional vector space, making it isomorphic to $V_{mn}(\kk)$. Our matrix representation of the tensor product of special symmetric spaces, $\rdottensor$, also satisfies this property when the vector spaces are properly identified---see~\Cref{lem:properties of rdottensor}.)

The inner products on the spaces $\mathrm{V}_m(\kk)$ and $\mathrm{V}_n(\kk)$ can be used to define a natural inner product on $\mathrm{V}_m(\kk) \otimes \mathrm{V}_n(\kk)$, turning it into a tensor product of inner product spaces. This is the notion of `tensor product' which is usually implied when dealing with Hilbert space (brushing aside all notions of completeness and topology due to working in finite dimensions). For any collection of elements $\{v_i\}, \{v_j'\} \subset \mathrm{V}_m(\kk)$, $\{w_i\}, \{w_j'\} \subset \mathrm{V}_n(\kk)$ and $\{a_i\},\{b_j\}\subset \kk$, define
\begin{equation}\label{eq:def:inner product tensor space}
	\left\langle
	\sum_i a_i v_i \otimes w_i , \,\sum_j b_j  v_j^{\prime} \otimes w_j^{\prime}\right\rangle 
	:= \sum_{i j} \overline{a}_i b_j\left\langle v_i \,,  v_j^{\prime}\right\rangle\left\langle w_i \, , w_j^{\prime}\right\rangle\:,
\end{equation}
where $\langle \cdot\, , \cdot \rangle$ denotes the inner product. It can be shown that the function so defined is a well-defined inner product. From this inner product, the inner-product space $\mathrm{V}_m(\kk) \otimes \mathrm{V}_n(\kk)$ inherits the other structure we are familiar with, such as the notions of an adjoint, unitarity, normality, and Hermiticity.

What sorts of linear operators act on the space $\mathrm{V}_m(\kk) \otimes \mathrm{V}_n(\kk)$? Suppose $ v$ and $ w $ are vectors in $\mathrm{V}_m(\kk)$ and $\mathrm{V}_n(\kk)$, and $A$ and $B$ are linear operators on $\mathrm{V}_m(\kk)$ and $\mathrm{V}_n(\kk)$, respectively. Then we can define a linear operator $A \otimes B$ on $\mathrm{V}_m(\kk) \otimes \mathrm{V}_n(\kk)$ by the equation
\begin{equation}
	(A \otimes B)( v \otimes w ) := A v \otimes B w \:.
\end{equation}

The definition of $A \otimes B$ is then extended to all elements of $\mathrm{V}_m(\kk) \otimes \mathrm{V}_n(\kk)$ in the natural way to ensure linearity of $A \otimes B$, that is,
\begin{equation}
	(A \otimes B)\left(\sum_i a_i v_i \otimes w_i\right) := \sum_i a_i \: [(A v_i)\otimes (B w_i)] \:.
\end{equation}

It can be shown that $A \otimes B$ defined in this way is a well-defined linear operator on $\mathrm{V}_m(\kk) \otimes \mathrm{V}_n(\kk)$, and thus can be extended uniquely to a linear operator on $\mathrm{V}_{mn}(\kk)$. This notion of the tensor product of two operators extends in the obvious way to the case where $A$ and $B$ map between different vector spaces e.g., $A: \mathrm{V}_m(\kk) \rightarrow \mathrm{V}_k(\kk)$ and $B: \mathrm{V}_n(\kk) \rightarrow \mathrm{V}_l(\kk)$. Indeed, an arbitrary linear map $C$ mapping $\mathrm{V}_m(\kk) \otimes \mathrm{V}_n(\kk)$ to $\mathrm{V}_k(\kk) \otimes \mathrm{V}_l(\kk)$ can be decomposed as a linear combination of tensor products of mappings $\mathrm{V}_m(\kk)$ to $\mathrm{V}_k(\kk)$ and $\mathrm{V}_n(\kk)$ to $\mathrm{V}_l(\kk)$,
\begin{equation}
	C=\sum_i c_i A_i \otimes B_i 
\end{equation}
where by definition
\begin{equation}
	\left(\sum_i c_i A_i \otimes B_i\right) v \otimes w  := \sum_i c_i \:[(A_i v) \otimes (B_i w)] \:.
\end{equation}

\subsection{Matrix representations of the tensor product}\label{sec:Matrix representations of the tensor-product}

A matrix representation of the tensor product is an explicit matrix combination rule that satisfies the above-stated properties. Crucially, it is not unique. A common choice in QT, is the so-called Kronecker product which we denote $\otimes$. 
\begin{definition}[Kronecker product]
	Let $A$ be an $m$-by-$n$ matrix and $B$ be a $p$-by-$q$ matrix, both with entries in the field $\kk$. Their Kronecker product is the $mp$-by-$nq$ matrix $A \otimes B$ defined by
	\begin{equation}
		A \otimes B := \overbrace{\left.\left[\begin{array}{cccc}
				A_{11} B & A_{12} B & \ldots & A_{1 n} B \\
				A_{21} B & A_{22} B & \ldots & A_{2 n} B \\
				\vdots & \vdots & \vdots & \vdots \\
				A_{m 1} B & A_{m 2} B & \ldots & A_{m n} B
			\end{array}\right] \right\} \!\!\!\!\!\!}^{\displaystyle n q} \,\, \,\, \,m p \:.
	\end{equation}
	Where terms like $A_{11} B$ in the above denote $p$ by $q$ submatrices whose entries are proportional to $B$, with overall proportionality constant $A_{11}$.
\end{definition}

The Kronecker product $\otimes$ can be proven to provide an essentially unique matrix representation of the tensor product of vector spaces whenever $\mathrm{V}_n(\kk) \cong \kk^n$. This is the case for RQT and CQT. 

For the RNQTs presented in the main text, since the involved vector spaces---the spaces of special symmetric matrices---are subspaces of matrix spaces, the Kronecker product is not guaranteed to be the right representation of the tensor product. As shown in \cref{sec:4bis_multipartite}, the right notion of the tensor product for these spaces happens to be the map $\rdottensor$, this is the construction considered in \cref{A new combination rule}. One can nevertheless force $\otimes$ to be used as the combination rule, if one is willing to drop the positivity of states and effects, this is the construction considered in \cref{sec:construction_Kronecker}.

\section{Construction of the mapping from CQT to RRQT}\label{app:construction}

We argued in the main text that locally mapping positive states and effects of CQT to positive states and effects of RQT (keeping postulates $(1)$ and $(2)$ of RQT) is in tension with mapping separable states and effects of CQT to separable states and effects of RQT (keeping postulate $(4)$ of RQT). In order to embed of CQT into a real theory that locally represent quantum theory (a RNQT), only one of these requirements can be satisfied at once.

In this appendix, we study the embedding and construct a new combination rule such that it can map separable states and effects to separable states and effects (i.e., we keep postulates $(1)$ and $(2)$ but consider the modified postulate $(4^\prime)$).

We start by recalling some definitions as well as facts on how to decompose matrices over the real and the complex numbers.

The $n$-dimensional Hilbert spaces over $\rr$ and over $\cc$, respectively noted as $\Hilb{}_\rr$ and $\Hilb{}_\cc$, can be respectively represented by the vector spaces $\rr^n$ and $\cc^n$. In such case, both share the standard basis $\{\ket{i} = (\delta_{i,j})_j\}_{i=0}^{n-1}$ as an orthonormal basis. The complex conjugation of a scalar $c$ will be noted with a bar like $\overline{c}$ and accordingly the complex conjugation of a vector $\ket{\psi} \in \cc^n$ with respect to the standard basis will be noted $\overline{\ket{\psi}}$. In the following, we will use the letter $\kk$ to refer to either $\rr$ or $\cc$. 

The elements of a Hilbert space of operators $\LinOp{\Hilb{}_\kk}$ of dimension $n$ are represented by square matrices in $\kk^{n\times n}$. For any matrix $M \in \kk$, the transpose with respect to the standard basis is noted $M^T$, whereas its Hermitian conjugate is noted $M^\dag := \overline{M}^T$. 
A matrix $M\in \kk^{n \times n}$ is called  \textit{Symmetric} iff $M^T = M$;  \textit{Antisymmetric} iff $M^T = -M$; \textit{Hermitian} iff $M^\dag = M$; and \textit{AntiHermitian} iff $M^\dag = -M$. 
Accordingly, we write $\SymX{n}{\kk}$, $\ASymX{n}{\kk}$, $\HermC{n}$, and $\AHermC{n}$ for the set of all $n \times n$ matrices that are, respectively, symmetric, antisymmetric, Hermitian, and antiHermitian. By the (real-)linearity of the transpose and the dagger, all of these sets are real subspaces of $\kk^{n \times n}$. Recall that the real vector space dimension of $\SymX{n}{\rr}$, $\ASymX{n}{\rr}, \SymX{n}{\cc}$, $\ASymX{n}{\cc}$, $\HermC{n}$, and $\AHermC{n}$ are, respectively, $\frac{n(n+1)}{2}, \frac{n(n-1)}{2}, n(n+1), n(n-1), n^2$ and $n^2$.

Notice that the spectrum of all matrices in $\SymX{n}{\rr}$ and in $\HermC{n}$ is purely real, whilst the spectrum of all matrices in $\ASymX{n}{\rr}$ and in $\AHermC{n}$ is purely imaginary. The former observation leads to the fact that the real algebra of self-adjoint operators on a Hilbert space defined over $\rr$ or $\cc$ are respectively represented by matrices in $\SymX{n}{\rr}$ or in $\HermC{n}$. Similarly, recall that the positive operators are represented by the corresponding sets of matrices with positive spectrum, that is, the sets of positive semi-definite matrices over $\rr$ and $\cc$ and that we respectively denote $\SymP{n}$ and $\HermP{n}$.

With a notion of transpose and Hermitian conjugate, matrices can be split into parts in a way that generalises the decomposition of a complex into its real and imaginary parts.
\begin{prop}\label{prop:CMatDecomposition}
	For any matrix $M \in \mathbb{C}^{n \times n}$ there exists four matrices $A,B,C,D \in \rr^{n\times n}$ such that 
	\begin{equation}\label{eq:CMatDecomposition_total}
		M = A + i B + i(C + iD) \:,
	\end{equation}
	with $A,C \in \SymX{n}{\rr}$ and $B, D \in \ASymX{n}{\rr}$.
\end{prop}
We delay the proof to introduce the following definitions:
\begin{definition}\label{def:CMatDecompositions}
	Using the elements of the decomposition \eqref{eq:CMatDecomposition_total} of a matrix $M$, define the...
	\begin{subequations}
		\begin{align}
			&\text{... real part of $M$ as} & \mathrm{Re}(M) := A - D \:; \\
			&\text{... imaginary part of $M$ as} & \mathrm{Im}(M) := B + C \:; \\
			&\text{... symmetric part of $M$ as} & \mathrm{Sym}(M) := A + iC \:; \\
			&\text{... antisymmetric part of $M$ as} & \mathrm{ASym}(M) := iB - D \:; \\
			&\text{... Hermitian part of $M$ as} & \mathrm{Herm}(M) := A + iB \:; \\
			&\text{... antiHermitian part of $M$ as} & \mathrm{Sym}(M) := i(C + iD) \:.
		\end{align}
	\end{subequations}
	This leads to the following decompositions of $M$:
	\begin{equation}
		M = \mathrm{Re}(M) + i \mathrm{Im}(M) = \mathrm{Sym}(M) + \mathrm{ASym}(M) = \mathrm{Herm}(M) + \mathrm{AHerm}(M) \:.
	\end{equation}
\end{definition}
\begin{corollary}
	From \cref{prop:CMatDecomposition} and \cref{def:CMatDecompositions},
	\begin{enumerate}
		\item The following relations hold:
		\begin{equation}
			\begin{array}{ccc}
				\overline{\mathrm{Re}(M)} = \mathrm{Re}(M) \:; & \mathrm{Sym}(M)^T = \mathrm{Sym}(M) \:; & \mathrm{Herm}(M)^\dag = \mathrm{Herm}(M) \:;   \\
				\overline{\mathrm{Im}(M)} = -\mathrm{Im}(M) \:; & \mathrm{ASym}(M)^T = - \mathrm{ASym}(M) \:; & \mathrm{AHerm}(M)^\dag = - \mathrm{Herm}(M) \:.
			\end{array}
		\end{equation}
		\item The following isomorphism of vector spaces hold:
		\begin{subequations}
			\begin{gather}
				\cc^{n\times n} \cong \rr^{n\times n} \oplus \rr^{n\times n} \cong \SymX{n}{\cc} \oplus \ASymX{n}{\cc} \cong \HermC{n} \oplus \AHermC{n}\:;\\
				\rr^{n\times n} \cong \HermC{n} \cong \SymX{n}{\rr} \oplus \ASymX{n}{\rr} \:.
			\end{gather}
		\end{subequations}
	\end{enumerate}
\end{corollary}
\begin{proof}
	The decompositions are explicitly constructed as:
	\begin{equation}
		\begin{array}{ccc}
			\mathrm{Re}(M)= \frac{M + \overline{M}}{2}\:; & \mathrm{Sym}(M)= \frac{M + M^T}{2}\:; & \mathrm{Herm}(M)= \frac{M + M^\dag}{2}\:; \\
			\mathrm{Im}(M)= \frac{M - \overline{M}}{2}\:; & \mathrm{ASym}(M)= \frac{M - M^T}{2}\:; & \mathrm{AHerm}(M)= \frac{M - M^\dag}{2}\:;
		\end{array}
	\end{equation}
	and
	\begin{equation}
		\begin{array}{cc}
			A = \mathrm{Sym}(\mathrm{Re}(M)) = \frac{\mathrm{Re}(M) + \mathrm{Re}(M)^T}{2} \:; & B = \mathrm{ASym}(\mathrm{Im}(M)) = \frac{\mathrm{Im}(M) - \mathrm{Im}(M)^T}{2} \:; \\
			C = \mathrm{Sym}(\mathrm{Im}(M)) = \frac{\mathrm{Im}(M) + \mathrm{Im}(M)^T}{2} \:; & D = -\mathrm{ASym}(\mathrm{Re}(M)) = -\frac{\mathrm{Re}(M) - \mathrm{Re}(M)^T}{2} \:. \\
		\end{array}
	\end{equation}
	The rest follows. Note that these constructions imply that only $M = \mathrm{Herm}{M} + \mathrm{AHerm}{M}$ is a basis-independent construction since it is based on the dagger, the other two rely on the transpose or the conjugate of a matrix, therefore they rely on the choice of basis in which these operations are taken.
\end{proof}

Finally, recall that the tensor products between pairs of real and complex matrices in, respectively, $\rr^{n\times n} \times \rr^{m\times m}$ and $\cc^{n\times n} \times \cc^{m\times m}$ are spanning the larger vector space, 
\begin{equation}
	\rr^{n\times n} \otimes \rr^{m\times m}\cong \rr^{nm \times nm}\:;\quad \cc^{n\times n} \otimes \cc^{m\times m} \cong \cc^{mn \times mn} \:.
\end{equation}
Here a notation like $\rr^{n\times n} \otimes \rr^{m\times m}$ means the tensor product of vectors spaces, meaning the vector space spanned by linear (with respect to the base field) combinations of tensor products of all pairs $(M_A,M_B) \in \rr^{n\times n} \times \rr^{m\times m} $; i.e.,
\begin{equation}
	\rr^{n\times n} \otimes \rr^{m\times m} := \mathrm{Span}_{\rr}\{ M_A \otimes  M_B | M_A \in \rr^{n\times n}, M_B \in \rr^{m\times m}\} \:.
\end{equation}
This property also holds for Hermitian matrices seen as a \textit{real} vector space,
\begin{equation}\label{eq:tensor_herm}
	\HermC{n} \otimes \HermC{m} \cong \HermC{nm} \:,
\end{equation}
as directly follows from dimensionality arguments; this is intrinsically the reason why complex quantum theory is locally tomographic. On the other hand, this property does not hold for the symmetric matrices nor the antisymmetric matrices. Rather, they respectively decompose as $\SymX{nm}{\rr} \cong (\SymX{n}{\rr} \otimes \SymX{m}{\rr}) \oplus (\ASymX{n}{\rr} \otimes \ASymX{m}{\rr})$ and $\ASymX{nm}{\rr} \cong (\SymX{n}{\rr} \otimes \ASymX{m}{\rr}) \oplus (\ASymX{n}{\rr} \otimes \SymX{m}{\rr})$. This has an important consequence for the following, namely that the space of tensor products of symmetric matrices is a proper subspace of the space of symmetric matrices on a tensor product space:
\begin{equation}
	\SymX{n}{\rr} \otimes \SymX{m}{\rr} \subset \SymX{nm}{\rr} \:, \quad n,m >1 \:.
\end{equation}

\subsection{Positivity-preserving mapping on a single system}\label{sec:4bis_single_system}
Let $\mathbb{1} = \begin{pmatrix}1 & 0 \\ 0 & 1\end{pmatrix}$ and $X,Y,Z$ be the Pauli matrices such that $Y = \begin{pmatrix}0 & -i \\ i & 0\end{pmatrix}$ and $J := XZ = -iY = \begin{pmatrix} 0 & -1 \\ 1 & 0 \end{pmatrix}$.

Let $\Gs$ be the mapping defined as in \cref{eq:local mapping} with the subscript referring to the dimension omitted for the whole section as we focus on a single system of dimension $d=n$. In this appendix, we extend by linearity the domain and codomain of $\Gamma_{d}$ from 
\begin{align}
	\HermC{n} \rightarrow \SY_n(\rr),
\end{align}
to
\begin{align}
	\cc^{n \times n} \rightarrow \rr^{2 n \times 2 n}.
\end{align}
That is, for any matrix $M\in\mathbb{C}^{n \times n}$ which decomposes as
\begin{equation}
	M=\text{Re}(M)+i\cdot\text{Im}(M)=(A-D)+i\cdot(B+C) \:,
\end{equation}
the map $\Gs: \cc^{n \times n} \rightarrow \rr^{2 n \times 2 n}$ reads
\begin{equation}
	M \mapsto \G{M}=\mathbb{1}\otimes\text{Re}(M)+J \otimes\text{Im}(M)=
	\begin{pmatrix}
		A-D & -(B+C) \\
		B+C & A-D
	\end{pmatrix}
\end{equation}

\begin{prop} \label{prop:ring}
	The map $\Gs$ is a $*$-ring homomorphism $(\mathbb{C}^{n\times n},+,\cdot, {~}^\dag) \rightarrow (\mathbb{R}^{2n\times 2n},+,\cdot, {~}^T)$. That is, for any two matrices $M,N \in \cc^{n \times n},$
	\begin{equation}\label{eq:homo_ring}
		\G{M+N}=\G{M}+\G{N}\:; \quad \G{MN}=\G{M}\G{N}\:; \quad \G{\id_n} = \id_{2n}\:; \quad \G{M^\dag} = \G{M}^T \:.
	\end{equation}
\end{prop}
\begin{proof}
	For any matrices $M,N \in \cc^{n \times n}$, showing that $\G{M+N}=\G{M}+\G{N}$ and that the identity is preserved is straightforward from the definition. The multiplication preservation $\G{MN}=\G{M}\cdot\G{N}$ follows from $\{\mathbb{1},J\}$ having the same algebraic rules as $\{1,i\}$ i.e., $J^2 = - \mathbb{1}$. Finally, the $*$-preservation part is proven via the decomposition: $\G{(M^\dag)}=\G{M}^T$ holds because on the one hand $\left(M^\dag=A+D+i\cdot(B-C)\right)$ while on the other hand,
	\begin{equation}
		\G{M}^T=
		\begin{pmatrix}
			(A-D)^T & (B+C)^T \\
			-(B+C)^T & (A-D)^T
		\end{pmatrix}
		=
		\begin{pmatrix}
			A+D & -(B-C) \\
			B-C & A+D
		\end{pmatrix} 
		= \mathbb{1} \otimes (A+D) + J \otimes (B-C) = \G{M^\dag}\:.
	\end{equation}
	This is because $A$ and $C$ are symmetric while $B$ and $D$ are antisymmetric.
\end{proof}

So this mapping sends the ring structure of complex matrices to the one of real matrices, and even preserves the adjoints. We can see by construction that it is injective but not surjective and that $\Gs$ will preserve any real linear combination. However, this mapping is not an algebra $*$-homomorphism. Indeed, it cannot properly map the scalars in $\cc^{n \times n}$ to those in $\rr^{n \times n}$ since it does not handle multiplication by a complex number as expected: $\G{(i\cdot M)}\neq i\cdot\G{M}$ as $i$ is not an element of $\rr$. Rather, multiplying any matrix $M$ by a complex $c \in \cc$ gives
\begin{equation}\label{eq:map_phase}
	\G{(c M)} = r \id \otimes (\cos{\theta} \ReP{M} - \sin{\theta} \ImP{M}) + r J \otimes (\cos{\theta} \ImP{M} +\sin{\theta} \ReP{M}) \:,
\end{equation}
for $c = r e^{i \theta}$, $r,\theta \in \rr$. 

Similarly, despite the preservation of adjoints, it will not be able to be an isometry with respect to the inner-product space structure induced by the Frobenius norms ($\TrX{}{\cdot^\dag \: \cdot}$ in $\cc^{n \times n}$ and $\TrX{}{\cdot^T \: \cdot}$ in $\rr^{2n \times 2n}$) because the base fields are different. Any complex-valued inner product will be mapped to 0 through the mapped inner product $\TrX{}{\G{\cdot}^T \: \G{\cdot}}$ since, $\forall M,N \in \cc^{n\times n}$,
\begin{equation}\label{eq:map_IP}
	\TrX{}{\G{N}^T \: \G{M}} = 2\: \ReP{\TrX{}{N^\dag\cdot M}} \:,
\end{equation}
as directly follows from $J$ having trace 0 and $J^2 = -\id$. 

Of course, requiring $\G{(i\cdot M)}=i\cdot\G{M}$ and $\TrX{}{N^\dag\cdot M} = \TrX{}{\G{N}^T \cdot \G{M}}$ is naive. The $\Gs$ mapping can actually be both a $*$-algebra homomorphism as well as an inner product space isomorphism by seeing the `base field' of $\mathbb{R}^{2m\times 2m}$ as the linear span of $\{\id,J\}$ in $\mathbb{R}^{2\times 2}$ instead of $\rr$. In this case, scalar multiplication would correspond to \cref{eq:map_phase} and the corresponding inner product would be such that $\ip{\G{N}}{\G{M}} = \TrX{}{N^\dag\cdot M}$ (up to a normalisation constant). But this exactly amounts to identifying $\cc$ as a $*$-algebra over $\rr$ with its representation in $\rr^{2 \times 2}$ given by $(1,i) \mapsto (\mathbb{1},J)$ such that (the notation at the space level is not rigorous but is there to help understanding the chain of equivalences involved):
\begin{equation}
	\begin{gathered}
		\cc^{n \times n} \cong \cc \otimes \rr^{n \times n} \cong (\rr \otimes \rr^{n \times n}) \oplus (\rr \otimes \rr^{n \times n})\rightarrow 
		\rr^{2 \times 2} \otimes \rr^{n \times n} \cong \rr^{2n \times 2n}  \\
		M = 1 \otimes \mathrm{Re}(M) + i \otimes \mathrm{Im}(M) \mapsto \mathbb{1}\otimes\text{Re}(M) + J \otimes\text{Im}(M)
	\end{gathered} \:.
\end{equation}
In that regard, one can summarize the objections with the Renou et al. approach presented in this paper as a review of the issues faced when considering $\rr$ to be the scalars of the real theory instead of the representation of $\cc$ in $\rr^{2\times 2}$. 
In the postulates, requiring every mention of $\cc$ to be replaced by $\rr$ instead of the real span of $\{\id,J\}$ causes them to lose structure when passing from complex quantum theory to their real theory, which explains their findings: if parts of the mathematical structure is lost when mapping from CQT to a real theory, it is expected that not all predictions of the former can be recovered in the latter. Using the particular mapping \cref{eq:local mapping} is just a way to pinpoint what has been lost under the change of representation, and how to fix it.

To start, it is necessary to first understand how the complex numbers are mapped under $\Gs$. To do so, we need some more notation: for $\id_n$ the $n \times n$ identity matrix and for $J$ as defined above, let
\begin{equation}
	J_n := J \otimes \id_n =
	\begin{pmatrix}
		0 & -\id_n \\
		\id_n & 0
	\end{pmatrix} \:.
\end{equation}
Similarly, let $I_n := \id \otimes \id_n$ such that $I_n = - J_n^2$. In the following, the $n$ index may be dropped when its presence should be clear from the context. With this notation, one sees that $I_n = \G{\id_n}$ and $J_n = \G{i\id_n}$; the matrices $I$ and $J$ represent the unit and imaginary scalars. 
Using the decomposition \eqref{eq:CMatDecomposition_total}, we have $iM = -(C+iD)+i(A+iB) = -(C+B) + i(A-D)$ and $\G{iM} = \begin{pmatrix}
	-(C+B) & -(A-D) \\
	A- D & -(C+B)
\end{pmatrix}$. 
This shows that 
\begin{equation}\label{eq:map1_linear}
	\forall a,b \in \rr, \quad \G{(a+bi) M} = a \G{M} + b J_n \G{M} \:,
\end{equation}
meaning that the $\Gs$ map preserves linear sums when the suitable $i \mapsto \G{i\id_n} = J_n$ has been applied. 
In particular, $\Gs$ becomes real-linear in the $b=0$ case. This situation corresponds to when the imaginary part of a complex matrix is zero, that is, when it is Hermitian. Remark that when required to work with Hermitian matrices only, the discrepancies between the image of complex quantum theory and the real theory, embodied in \cref{eq:map_phase,eq:map_IP}, disappear. As postulates (1) to (3) tell us that the predictions of quantum theory are inner products of positive, thus Hermitian, matrices, it appears plausible that the mathematical structure lost when passing from (1) and (2) to (1$^\prime$) and (2$^\prime$) has no observable consequences. We now prove this claim.

In the main text, we name the image of Hermitian matrices through $\Gamma$ the special symmetric matrices. 
\begin{definition} The Special Symmetric Matrices are the matrices $M\in\mathbb{R}^{2n\times 2n}$ which have the form 
	\[
	M=\begin{pmatrix}
		A & -B \\
		B & A
	\end{pmatrix},
	\]
	for $ A\in\SymX{n}{\rr}$ and $B\in\ASymX{n}{\rr}$. The space of all such matrices is noted $\SYR{n}$.
\end{definition}
An interesting observation of St\"{u}kleberg was that the special symmetric matrices are precisely the symmetric matrices that behave correctly with the representation of the imaginary unit, meaning that they commute with it~\cite{Stueckelberg1960}.
\begin{lemma}\label{lem:SYM_char}
	\begin{equation}
		\G{M} \in \SYR{n} \iff \G{M} \in \SymX{2n}{\rr} \: \& \: J_n \G{M} = \G{M} J_n \:.
	\end{equation}
\end{lemma}
From the definition, it is obvious that $\SYR{n}$ is a subspace of $\rr^{2n \times 2n}$ of dimension $n^2$. It is also direct to see that the image of $\HermC{n}$ through the $\Gs$ map is contained in $\SYR{n}$. From there, one can infer the $\Gs$ map to be an isomorphism between $*$-algebras over $\rr$.
\begin{prop}\label{prop:algebra}
	The map $\Gs$ is an isomorphism of $*$-algebras over $\rr$ between $\HermC{n}$ and $\SYR{n}$. That is, in addition to the ring structure obtained from \Cref{prop:ring}, the following holds for any two Hermitian matrices $H,L$ and reals $a,b$:
	\begin{equation}
		\G{aH + bL} = a\G{H} + b\G{L} \:.
	\end{equation}
\end{prop}
\begin{proof}
	The only missing part was the linearity w.r.t. real coefficient of $\Gs$ when applied to Hermitian matrices. It thus follows from taking \cref{eq:map1_linear} in the case where $b=0$.
\end{proof}

Focusing on Hermitian matrices is relevant not only because they form a real algebra. Indeed, what is important in a theory is its predictive power. Yet, the main quantities that are lost when considering the representation of CQT over the reals are the complex-valued inner products. However, they should have little impact on the predictive power of the theory, as all quantum states and effects are positive operators, leading to positive-valued inner products. While it is true, as we will now show, the loss of the complex-algebraic structure will still have an influence, but a more subtle one. 

\begin{prop}\label{prop:isometry}
	The map $\Gs$ when applied from $\HermC{n}$ to $\SYR{n}$ is, up to renormalization of the trace by a factor of 1/2, an isometry of inner product spaces.
\end{prop}
\begin{proof}
	It follows from \cref{eq:map_IP}, but we will prove this formula using the decomposition \cref{eq:CMatDecomposition_total} for completeness.  
	
	For two such decomposed matrices $M = (A + iB) + i(C +iD) $ and $N = (A' + iB') + i(C' + iD')$, it can be computed that $\TrX{}{M^\dag N} = (A-D)(A'-D')-(B-C)(B'+C')+i\left((B-C)(A'-D') + (A-D)(B'+C')\right)$ and $\TrX{}{\G{M}^T\G{N}} = 2\left((A-D)(A'-D')+(B+C)(B'+C')\right)$. 
	
	$M$ and $N$ are made Hermitian setting $C,D,C'$ and $D'$ to 0, In such a case, the mapping is one-to-one, $\G{\HermC{n}} = \SYR{n}$ and the above computation leads to
	\begin{equation}
		\forall H,L \in \HermC{n}, \quad  \TrX{}{\G{H} \G{L}} = 2 \TrX{}{H L} \:,
	\end{equation}
	showing that $\Gs$ is an isometry (up to a factor of 2) when applied to Hermitian matrices.
\end{proof}

Hence, passing from $\HermC{n}$ to $\SYR{n}$ under $\Gs$ is more sensible than from $\cc^{n \times n}$ to $\rr^{n \times n}$ as less structure is lost in the process. We can now focus on the set of states, and for that we need to track what happens to the trace as well as the positivity.

For the positivity condition, we look at the image through $\Gs$ of the set $\HermP{n}$ of the $n\times n$ positive semi-definite matrices over $\cc$.
\begin{theorem}\label{thm:posvitive semi definitness conservation of the local mapping}
	Under the mapping $\Gs$, the positive semi-definite matrices in $\HermC{d}$ are mapped to positive semi-definite matrices in $\SYR{n}$.
\end{theorem}
\begin{proof}
	Let $M \in \cc{n \times n}: \: M \geq 0$ meaning that $z^\dag M z \geq 0$ $\forall z \in \cc^n$. Then
	\begin{equation}
		\begin{aligned}
			\forall z: \quad &z^\dag M z \geq 0 \iff \\
			\forall x,y \in \rr^n, \forall A \in \SymX{n}{\rr}, \forall B \in \ASymX{n}{\rr}: \quad & (x - i y)^T (A + iB) (x + iy) \geq 0 \iff\\
			x^T A x + y^T A y + 2 y^T B x \geq 0 \:.
		\end{aligned}
	\end{equation}
	The last line has been reached by distributing the parentheses, then using $x^T M y = y^T M^T x$, which implies that $x^T B x = 0$ for $B$ antisymmetric and $x^T A y - y^T A x = 0$ for $A$ symmetric.
	
	Now let $\G{M} \in \rr{2n \times 2n}: \: \G{M} \geq 0$ meaning that $\G{z}^T M \G{z} \geq 0$ $\forall \G{z} \in \rr^{2n}$. Then,
	\begin{equation}
		\begin{aligned}
			\forall \G{z} : \quad &\G{z}^T M \G{z} \geq 0 \iff \\
			\forall x,y \in \rr^n, \forall A \in \SymX{n}{\rr}, \forall B \in \ASymX{n}{\rr}: \quad & \begin{pmatrix}
				x \\ y
			\end{pmatrix}^T \begin{pmatrix}
				A & -B \\ B & A
			\end{pmatrix} \begin{pmatrix}
				x \\ y
			\end{pmatrix}	 \geq 0 \iff\\
			x^T A x + y^T A y + 2 y^T B x \geq 0 \:,
		\end{aligned}
	\end{equation}
	so we have reached the same condition as the above one. This should be no surprise from $\Gs$ being an algebra $*$-homomorphism and the variables have been chosen to make it clear that the second condition is but the image of the first under $\Gs$.
\end{proof}
For the trace, we have just shown that $\text{Tr}[\G{H}\G{L}]=2\text{Tr}[HL]$ for $H$ and $L$ Hermitian. Taking $L=\id_n$, we recover the definition of the $M$ map in the main text, \cref{eq:map M}, for a single subsystem and a prefactor of $a_1=1/2$.

Therefore, the states in $\text{Herm}_n(\mathbb{C})$ are in one-to-one correspondence with states in $\SYR{n}$ under $M$. As for the effects, since the $\Gs$ map preserves inner products and positivity, it preserves duals as well so they also correspond. This also hold if the effects are restricted to projectors, whose image through $\Gs$ are also proportional to projectors by consequence of \cref{eq:homo_ring}.

We now make a short detour to investigate what are the matrices outside of the image of Hermitian matrices.
\begin{lemma}\label{lem:ASYM_char}
	\begin{equation}
		A \in \SymX{2n}{\rr}\setminus \SYR{n} \iff A \in \SymX{2n}{\rr} \: \& \: J_n A = -A J_n \:.
	\end{equation}
\end{lemma}
\begin{proof}
	Starting from a symmetric matrix $S \in \SymX{2n}{\rr}$, one defines $S^- := \frac{S - JSJ}{2}$ and $S^+ := \frac{S - JSJ}{2}$. Remark that $(JSJ)^T=JSJ$ so both matrices are symmetric. Next, remark that $JS^- = JS + SJ = S^-J$ so $S^-$ commutes with $J$ and so $S^+$ anticommutes with $J$. Using \Cref{lem:SYM_char}, it must be that $S^- \in \SymX{n}{\rr}$ which implies $S^+ \in \SymX{2n}{\rr}\setminus \SYR{n}$. Therefore, we have proven the right implication in the above formula. Finally, the converse also holds because \Cref{lem:SYM_char} is a necessary and sufficient condition.
	\end{proof}
	\begin{corollary}
		The set of symmetric states that are not special symmetric are orthogonal to the special symmetric ones:
		\begin{equation}
			\forall M \in \SYR{n},\: \forall A \in \SymX{2n}{\rr}\setminus \SYR{n},\quad \TrX{}{M A} = 0 \:.
		\end{equation}
		As a consequence, these matrices are traceless. 
	\end{corollary}
	\begin{remark}
		Notice that the state space of the real quantum theory considered in Ref.~\cite{Renou2021} is larger than a trace-normalized subset of $\SYRP{n}$, as it is taken to be one of $\SymP{2n}$ and $\SYRP{n} \subset \SymP{2n}$. 
		In the case of the map $\Gs$, this is not a problem to restrict the state space to $\SYRP{n}$, as the ``states'' in $\SymP{2n}\setminus \SYRP{n}$ are orthogonal to those in $\SYRP{n}$. 
		
		Provided that the dimension is at least $2n$, the corresponding CQT can be represented in $\SYRP{n}$, and the extra states in $\SYRP{n} \subset \SymP{2n} $ cannot prevent the obtainable probability distributions to be at least those obtainable with the CQT.
	\end{remark}
	
	Going back to the action of $\Gs$ on complex quantum theory, the last aspect of the theory that has not been addressed is its dynamics. We start by considering the evolution of a state $\rho$ under unitary dynamics $U$, $\rho \mapsto U \rho U^\dag$. Because $\Gs$ is an algebra homomorphism, the action of a unitary transform decomposes into $\G{U \rho U^\dag} = \G{U}\G{\rho}\G{U}^T$. We already know what is the image of the set of valid states, so what remains is to characterize the image of unitary matrices through $\Gs$.
	\begin{definition} 
		A map $S\in\mathbb{R}^{2n\times 2n}$ is
		\[
		\begin{aligned}
			&\text{symplectic when} && S^T J = J S^{-1} \:;\\
			&\text{antisymplectic when} && S^T J = -JS^{-1}.
		\end{aligned}
		\]
	\end{definition}
	\begin{prop}\label{prop:image_unitary}
		If $U\in\mathbb{C}^{n\times n}$ is unitary then $\G{U}\in\mathbb{R}^{2n\times 2n}$ is orthogonal \& symplectic. This implies the following characterization:
		\begin{equation}
			U^\dag={U}^{-1} \quad \iff \G{U}^T = \G{U}^{-1} \quad\&\quad \G{U} J=J\G{U} \:.
		\end{equation}
		Such matrix will be referred to as \textit{orthosymplectic}. \\
		
		If $V\in\mathbb{C}^{n\times n}$ is antiunitary, meaning that $V = UK$ with $K$ the complex conjugation map $K\ket{\psi} = \ket{\overline{\psi}}$ $\forall \psi$ then $\G{V}\in\mathbb{R}^{2n\times 2n}$ is orthogonal \& antisymplectic. This implies the following characterization:
		\begin{equation}\label{eq:antiortho}
			V=UK:\quad U^\dagger=U^{-1} \quad\&\quad \forall c \in \cc :\: Kc = \overline{c} \quad \iff\quad\G{V}^T=\G{V}^{-1}\;\&\;\G{V} J=-J\G{V} \:.
		\end{equation}
	\end{prop}
	See Ref.~\cite{Myrheim1999} for a proof.
	
	\begin{remark}
		As a side note, and building upon an observation of Myrheim~\cite{Myrheim1999}, it would be interesting to use this RNQT as a means to have a linear representation of time-reversal in quantum theory. Remember that due to Wigner theorem, pure dynamics in quantum theory can only be unitary or antiunitary, with the antiunitary transformations corresponding to unitary transformations followed by a time-reversal. Collecting some results, we see the potential of the real quantum theory: antiunitary transformations have a real-linear representation in RNQT as antiorthosymplectic matrices (characterized by \cref{eq:antiortho}) and the (unormalized) time-reversed states are exactly those in $\SymP{n}\setminus\SYR{n}$ (characterized by \Cref{lem:ASYM_char}), since they are the image of quantum states that have undergone an antiunitary transform. 
	\end{remark}
	As an example of antiunitary transformation representable in RNQT, take $\overline{M} = A -iB - iC - D$, one computes that this antiunitary is now but a $\frac{\pi}{2}$ rotation of the $\rr^{2\times2}$ matrices used to represent the complex scalars, $\G{\overline{M}} = Z_n \G{M} Z_n$. Here $Z_n = Z \otimes \id$ where $Z$ is the Pauli $Z$ matrix.
	
	With the discussion so far, one may wonder when the predictions of real quantum theory may diverge from its complex counterpart. 
	No big discrepancy has been found so far, except that the non-physical antiunitary transformation admit a linear representation. Still, this is also feasible in mixed state representation of CQT via the transposition. So where do the theories diverge? The answer depends on how the combination rule is defined. Many authors in the 2000s started considering the tensor product $\otimes$ for combination in RQT by formal analogy with CQT. However, a mapping from CQT to RQT that preserves the tensor product will face an issue: this product is bilinear in the base field. We have seen already that trying to map the scalars in $\cc$ to those in $\rr$ result in loosing structure; one should instead map the base field $\cc$ to its representation in $\rr^{2 \times 2}$, which is the real span of $\{\id,J\}$ and look for a combination rule that is bilinear with respect to this representation. 
	Yet, we have also seen above that ignoring this mathematical subtlety leads to no issue when restricting oneself to Hermitian operators: the inner products are preserved under $\Gs$ and thus they can be no discrepancies between the probability distributions obtained from each model. We will see in the next section if this is still the case.
	
	\subsection{Positivity-preserving mapping on two subsystems}\label{sec:4bis_multipartite}
	So far, the theory relied on representing $\cc$ as matrices over the reals, with the map $\Gs$ being the representation from the complex vectors to the general linear group on $\rr^{2n}$  Defining the action of the mapping over multiple subsystems in tensor product then amount to taking the tensor product of representations.
	But it is well-known that the tensor product of two representations is rarely an irreducible representation. The approach taken in this section can then be seen as a Clebsch-Gordon procedure: we will try to build an irreducible representation out of the tensor product one. 
	
	We start with the bipartite case. In terms of spaces isomorphisms, the issue can be phrased in symbols as 
	\begin{equation}
		\begin{array}{lclclcrclcrclclcl}
			\cc^{n \times n} &\otimes &\cc^{m \times m} &\cong (&\cc &\otimes &\rr^{n \times n} ) &\otimes (&\cc &\otimes &\rr^{m \times m}  ) &\cong &\cc &\otimes &\rr^{mn \times mn} & \cong & \cc^{mn\times mn}  \\
			\downarrow\Gs_n && \downarrow\Gs_n && \downarrow\Gs_1 &&&& \downarrow\Gs_1 &&&& \downarrow\Gs_{1\times1} &&&& \downarrow\Gs_{n\times m}\\
			\rr^{2n \times 2n} &\otimes &\rr^{2m \times 2m} &\cong ( &\rr^{2\times 2} & \otimes &\rr^{n \times n} )  &\otimes  (& \rr^{2\times 2} &\otimes &\rr^{m \times m}) &\cong  &\rr^{4\times4} &\otimes &\rr^{mn \times mn} &\cong & \rr^{4mn \times 4mn}
		\end{array}\quad.
	\end{equation}
	And so we are trying to find a way to compose $\Gs_n$ with $\Gs_m$ such that it returns some map $\Gs_{n\times m}$ derived out of $\Gs_n \otimes \Gs_m$ and which provides an irreducible representation of the underlying complex quantum theory.
	
	There is a way to directly see that the tensor product of two representations will not be irreducible: since the tensor product $\otimes$ is a bilinear map over different fields ($\cc$ and $\rr$), there is a phase invariance that will be problematic to pass from one representation to the other. That is, the map $\Gs_n \otimes \Gs_m$ will have trouble preserving the equality $M_A \otimes M_B = (e^{i\theta} M_A) \otimes (e^{- i \theta}M_B)$ for two arbitrary matrices $M_A \in \cc^{n \times n}$ and $M_B \in \cc^{m \times m}$ since it is not complex-bilinear. 
	Compare $\Gi{n}{M_A} \otimes \Gi{m}{M_B} $ to $ \Gi{n}{e^{i\theta} M_A} \otimes \Gi{m}{e^{- i \theta}M_B}$; from the algebra, these two expressions are equivalent when $\Gi{n}{e^{i\theta}\id_n} \otimes \Gi{m}{e^{-i \theta}\id_m} = I_n \otimes I_m$, i.e., $\Gi{n}{e^{i\theta}\id_n} \otimes \Gi{m}{e^{-i \theta}\id_m} = \Gi{n}{\id_n} \otimes \Gi{m}{\id_m}$. Here we see that the real-linearity of $\Gs$ induces an issue: $\Gs_n \otimes \Gs_m$ cannot be an isomorphism as, for example, if $\theta=\pi/2$ then $\id_n \otimes \id_m = -(i\id_n)\otimes (i\id_m)$ but $\Gi{n}{\id_n} \otimes \Gi{m}{\id_m} \neq -\Gi{n}{i\id_n} \otimes \Gi{m}{i\id_m}$ since $I_n \otimes I_m \neq -J_n \otimes J_m$.
	Even worse, it cannot be an isometry, as it fails to preserve traces. In $I_n \otimes I_m \neq -J_n \otimes J_m$, the left-hand side has a trace of $4nm$ while the right-hand side is traceless, for instance. Actually, since $\TrX{}{J_n \Gi{n}{M_A}} = 0$ for any Hermitian matrix, the equality $\TrX{}{M_A \otimes M_B} = - \TrX{}{iM_A \otimes iM_B}$ fails to be reproduced through $\Gs_n \otimes \Gs_m$: $ \TrX{}{\Gi{n}{M_A} \otimes \Gi{m}{M_B}} = \TrX{}{\Gi{n}{iM_A} \otimes \Gi{m}{iM_B}}$ would imply $\TrX{}{M_A}\TrX{}{M_B} = 0$ for any pair of Hermitian matrices, which cannot be true.
	
	Generalizing for any matrices $M_A,N_A \in \cc^{n \times n}$ and $M_B,N_B\in \cc^{m \times m}$, one has
	\begin{equation}
		0 = \TrX{}{(\G{M_A} \otimes \G{M_B})^\dag (\G{iN_1}\otimes \G{iN_2})} \neq -\TrX{}{(M_A \otimes M_B)^\dag(N_1 \otimes N_2)} \:.
	\end{equation}
	The underlying reason behind this issue is the following: trying to define $\Gs_{n \times m}$ using the tensor product is that the tensor product of the two copies of the representation of $\cc$ in $\rr^{2 \times 2}$ induces a representation of $\cc$ in $\rr^{4 \times 4}$ that is reducible; it has induced a two-fold degeneracy $1 \rightarrow \{I_n \otimes I_m, -J_n \otimes J_m\}$ and $i \rightarrow \{J_n \otimes I_m, I_n \otimes J_m\}$ under which $\G{\cdot} \otimes \G{\cdot}$ and $-\G{i\cdot}\otimes \G{i\cdot}$ provide the two orthogonal representations. 
	
	Rather, we would have preferred a map $\Gs_{n \times m}$ which directly sent the two spaces to be composed into the irreducible $4nm$-dimensional representation $1 \rightarrow I_{n\times m}$ and $i\rightarrow J_{n \times m}$. To do so, the standard procedure is to mix over the irreps:
	\begin{equation}\label{eq:2mix}
		\G{M_A} \rwedgetensor \G{M_B} := \frac{1}{2}\left[\G{M_1} \otimes \G{M_2} - \G{iM_1} \otimes \G{iM_2}\right] \:.
	\end{equation}
	As we will see, this procedure solves every representation-induced problem. But to grab its content more easily one needs some reshuffling and a detour in stabilizer state formalism. 
	
	The first thing we will do is to reshuffle the tensor factors ordering for easier-to-read equations. In \cref{eq:2mix} above, it is assumed that the underlying space of matrices factors as $\rr^{4n \times 4n} = \rr^{2\times 2}_A \otimes \rr^{n \times n}_A \otimes \rr^{2\times 2}_B \otimes \rr^{m \times m}_B$, where we put mention to $A$ and $B$ to remind which space is attached with whom. Until the end of this appendix, we take the convention to put all the `scalar' spaces on the left while retaining the parties' ordering. For the above example this means putting the two $\rr^{2 \times 2}$ spaces where the complex numbers are represented on the left of the factorization, $\rr^{4n \times 4n} = \rr^{2\times 2}_A \otimes \rr^{2\times 2}_B \otimes \rr^{n \times n}_A \otimes \rr^{m \times m}_B$, so swapping the second and third factor. 
	
	We will need the following $4$-dimensional representation of $\{1,i\}$ in $\rr^{4 \times 4}$:
	\begin{subequations}\label{eq:IJ_4d}
		\begin{align}
			& \I{2} := \frac{1}{2} \left( \id \otimes \id - J \otimes J \right) \:; \\
			& \J{2} := \frac{1}{2} \left( \id \otimes J + J \otimes \id \right) \:.
		\end{align}
	\end{subequations}
	As well as the expression of the real and imaginary parts of a tensor product, recall that 
	\begin{subequations}\label{eq:Re_Im_bipartite}
		\begin{align}
			& \ReP{M_A \otimes M_B} = \ReP{M_A} \otimes \ReP{M_B} - \ImP{M_A} \otimes \ImP{M_B} \:;\\
			& \ImP{M_A \otimes M_B} = \ReP{M_A} \otimes \ImP{M_B} + \ImP{M_A} \otimes \ReP{M_{B}} \:.
		\end{align}
	\end{subequations}
	Using these four substitutions, a long but straightforward computation yields
	\begin{equation}\label{eq:2mix_Re_Im}
		\G{M_A} \rwedgetensor \G{M_B} = \I{2} \otimes \ReP{M_A \otimes M_B} + \J{2} \otimes \ImP{M_A\otimes M_B}\:.
	\end{equation}
	One formally recovers the formula \cref{eq:map local} in the above! There is a reason for that. Of course, the pair $\{ \I{2}, \J{2} \}$ is indeed a representation of $\{1,i\}$ in $\rr^{4 \times 4}$ as it can be checked that they obey the proper multiplication rules, $\I{2} = \big(\I{2}\big)^2 = -\big(\J{2}\big)^2$ and $\J{2} = \I{2}\J{2} = \J{2}\I{2}$, that the transpose plays the role of conjugation: $\big(\I{2}\big)^T = \I{2}$ and $\big(\J{2}\big)^T = - \J{2}$, and that they are orthogonal to each other, $\TrX{}{\big(\I{2}\big)^T \cdot \J{2}} = 0$.
	But what is useful with this representation is that it handles local application of smaller-dimensional representations properly:
	\begin{equation}\label{eq:IJ_4d_local}
		(\id \otimes J) \I{2} = \I{2} (\id \otimes J) = (J \otimes \id) \I{2} = \I{2} (J \otimes \id) \
		= \J{2}\I{2} \:.
	\end{equation}
	This property is what makes these representations useful in stabilizer state formalism. For our purpose, it makes the $\rwedgetensor$ behave properly with respect to the phase freedom of the underlying complex quantum theory.
	To prove it, one should first show that the $\rwedgetensor$ product is \textit{robust} against the tensor product in the following sense:
	\begin{lemma}\label{lem:2mix_absorbant}
		Let $M_A \in \cc^{n\times n}$ and $M_B \in \cc^{m \times m}$ be matrices, let $I_d = \id \otimes \id_d$ where $\id$ is the 2-by-2 identity matrix and $\id_d$ its $d$-by-$d$ counterpart, and let $\Gs_n$ and $\Gs_m$ as defined above. Then, the following relations hold
		\begin{subequations}\label{eq:2mix_absorbant}
			\begin{align}
				&\Gi{n}{M_A} \rwedgetensor \Gi{m}{M_B} = \left(I_n \rwedgetensor I_m\right)\left(\Gi{n}{M_A} \otimes \Gi{m}{M_B}\right) \\
				&\Gi{n}{M_A} \rwedgetensor \Gi{m}{M_B} = \left(\Gi{n}{M_A} \otimes \Gi{m}{M_B}\right)\left(I_n \rwedgetensor I_m\right)
			\end{align}
		\end{subequations}
	\end{lemma}
	\begin{proof}
		Starting from the definition of the left-hand side,
		\begin{equation}
			\begin{aligned}
				\Gi{n}{M_A} \rwedgetensor \Gi{m}{M_B} =& \frac{1}{2}\left( \Gi{n}{M_A} \otimes \Gi{m}{M_B} - \Gi{n}{i M_A} \otimes \Gi{m}{i M_B} \right)\\
				=& \frac{1}{2}\left( \Gi{n}{M_A} \otimes \Gi{m}{M_B} - (J_n\Gi{n}{M_A}) \otimes (J_m\Gi{m}{M_B}) \right)\\
				=& \frac{1}{2}\left( I_n \otimes I_m  - J_n \otimes J_m \right)(\Gi{n}{M_A} \otimes \Gi{m}{M_B})\\
				=& \frac{1}{2}\left( (\id \otimes \id - J \otimes J ) \otimes (\id_n \otimes \id_m) \right)(\Gi{n}{M_A} \otimes \Gi{m}{M_B})\\
				=& \left( \I{2} \otimes (\id_n \otimes \id_m) \right)(\Gi{n}{M_A} \otimes \Gi{m}{M_B})\\
				&\overset{\eqref{eq:2mix_Re_Im}}{=} \left( I_n \rwedgetensor I_m \right)(\Gi{n}{M_A} \otimes \Gi{m}{M_B}) \:.
			\end{aligned}
		\end{equation}
		In the above, the refactorisation announced earlier in the appendix has been performed when going to the third line, and so equalities should be understood up to the insertion of an ad hoc swap operator. The second statement follows directly by the same reasoning, and that $\Gi{n}{i M_A} = \Gi{n}{M_Ai} = \Gi{n}{M_A} J_n$.
	\end{proof}
	
	Notice that the previous proof was made without needing the underlying complex matrix, only that the matrices were commuting with $J_n \otimes J_m$. Hence, this result holds for any pair of matrices in $\rr^{2n\times 2n} \times \rr^{2m \times 2m}$ whose tensor product commutes with $J_n \otimes J_m$. In general, any result proven for the image of some complex matrix through $\Gs$ will also hold for an arbitrary real matrix as long as it commutes with $J$. 
	
	We are now in a position to prove that the new bilinear combination rule properly accounts for the phase freedom.
	\begin{prop}
		The real bilinearity of the $\rwedgetensor$ combination properly renders the complex bilinearity of the complex tensor product $\otimes_\cc$ through the mapping $\Gs$ in the following sense: $\forall \theta\in \rr, M \in \cc^{n \times n}, N\in \cc^{m \times}$,
		\begin{equation}
			\Gi{n}{e^{i\theta}M } \rwedgetensor \Gi{m}{N} = \Gi{n}{M } \rwedgetensor \Gi{m}{e^{i\theta} N} = \cos{\theta}(\Gi{n}{M} \rwedgetensor \Gi{m}{N}) + \sin{\theta} (\Gi{n}{iM} \rwedgetensor \Gi{m}{N}) \:.
		\end{equation}
		In particular, the phase freedom of $\otimes_\cc$ is preserved through $\Gs$:
		\begin{equation}
			\Gi{n}{M} \rwedgetensor \Gi{m}{N}   =  \Gi{n}{e^{i\theta} M} \rwedgetensor \Gi{m}{e^{-i\theta}N}
		\end{equation}    
	\end{prop}
	\begin{proof}
		The equivalence between the first and last formulae, $\Gi{n}{e^{i\theta}M } \rwedgetensor \Gi{m}{N} = \cos{\theta}(\Gi{n}{M} \rwedgetensor \Gi{m}{N}) + \sin{\theta} (\Gi{n}{iM} \rwedgetensor \Gi{m}{N})$ is direct by bilinearity. Ti show equivalence with the second formula, we use \Cref{lem:2mix_absorbant} as well as $\Gs_{n}$ being an homomorphism, this gives $\Gi{n}{e^{i\theta}M } \rwedgetensor \Gi{m}{N} = (\Gi{n}{e^{i\theta} \id_n} \rwedgetensor \Gi{m}{\id_m})(\Gi{n}{M } \otimes  \Gi{m}{N})$. It remains to show that the phase can pass on the right-hand side of the $\rwedgetensor$ in $(\Gi{n}{e^{i\theta} \id_n} \rwedgetensor \Gi{m}{\id_m})$. By successively using \cref{eq:2mix_Re_Im} and \cref{eq:IJ_4d_local} we have
		\begin{equation}
			\begin{aligned}
				\Gi{n}{e^{i\theta} \id_n} \rwedgetensor \Gi{m}{\id_m} 
				=& \cos{\theta}(\Gi{n}{\id_n} \rwedgetensor \Gi{m}{\id_m}) + \sin{\theta} ({J_n}\Gi{n}{\id_n} \rwedgetensor \Gi{m}{\id_m}) \\
				=& \cos{\theta}( \I{2} \otimes (\id_n \otimes \id_m)) + \sin{\theta} ({J_n} \otimes I_m) ( \I{2} \otimes (\id_n \otimes \id_m) )\\
				=& \cos{\theta}( \I{2} \otimes (\id_n \otimes \id_m) ) + \sin{\theta} (I_n \otimes J_m) ( \I{2} \otimes (\id_n \otimes \id_m) )\\
				=& \cos{\theta}(\Gi{n}{\id_n} \rwedgetensor \Gi{m}{\id_m}) + \sin{\theta} (\Gi{n}{\id_n} \rwedgetensor J_m\Gi{m}{\id_m}) \\
				&=\Gi{n}{ \id_n} \rwedgetensor \Gi{m}{e^{i\theta}\id_m}
			\end{aligned}\quad .
		\end{equation}
		The rest follows.
	\end{proof}
	
	So the $\rwedgetensor$ is free of the phase issue faced by the $\otimes$. Combining the mixed product property with \Cref{lem:2mix_absorbant} shows that the $\rwedgetensor$ product will `absorb' any tensor product it encounters, so it is replacing it in every formula once it is introduced. 
	Nonetheless, it is doing it in a fashion that does not change the traces of Hermitian matrices,
	\begin{equation}\label{eq:rwedgetensor_trace}
		\TrX{}{\Gi{n}{M_A} \rwedgetensor\Gi{m}{M_B}} = \TrX{}{\Gi{n}{M_A} \otimes \Gi{m}{M_B} } = 2^2 \TrX{}{M_A}\TrX{}{M_B} \:,
	\end{equation}
	as can be computed using the definition \eqref{eq:2mix}. 
	
	But its robustness and it fixing the phase issue are only qualitative reasons why $\rwedgetensor$ is a better combination rule than $\otimes$. These are model-dependent arguments and do not tell us much about the achievable distributions. There is nonetheless a `better' choice of combination rule.
	
	As we mentioned in the previous section, on single-partite systems, the $\Gs$ map is a sensible choice of mapping between complex and real theories as it injectively maps the complex states and effects spaces into their real counterparts. Only with respect to that, the switch from postulates (1) and (2) to (1') to (2') makes sense.
	Now we may wonder if it is still the case in light of postulate (4): are composite state spaces of complex quantum theory embedded in those of the real theory through \textit{local application} of the map $\Gs$?
	
	A first element of answer lies in the property $\eqref{eq:tensor_herm}$: the real span of a composite complex state space can be spanned by its tensor products. It will not be the case for its image as $\Gi{nm}{\HermC{nm}} \subset \rr^{2nm \times 2nm}$ whereas $\Gi{n}{\HermC{n}} \rtensor \Gi{m}{\HermC{m}} \subseteq \rr^{4nm \times 4nm}$ for $\rtensor = \{\otimes,\rwedgetensor\}$; the dimensions do not match. But this dimensionality issue should be seen as a speck of the representation rather than an obstacle. To go beyond it, one can still wonder if they are still isomorphic as state spaces, i.e. whether $\Gi{nm}{\HermC{nm}} \cong \Gi{n}{\HermC{n}} \rtensor \Gi{m}{\HermC{m}}$, which is the case.
	\begin{prop}\label{prop:iso_spaces}
		For $\rtensor = \{\otimes,\rwedgetensor\}$, the real vector space $\Gi{n}{\HermC{n}} \rtensor \Gi{m}{\HermC{m}}$ is a subspace of the special symmetric matrices in $\rr^{4nm \times 4nm}$, i.e.,
		\begin{subequations}
			\begin{align}
				&\Gi{n}{\HermC{n}} \otimes \Gi{m}{\HermC{m}} \subseteq \SYR{2nm} \:;\\
				&\Gi{n}{\HermC{n}} \rwedgetensor \Gi{m}{\HermC{m}} \subseteq \SYR{2nm} \:.
			\end{align}
		\end{subequations}
		Moreover, both subspaces are isomorphic to $\SYR{nm}$.
	\end{prop}
	\begin{proof}
		Since the transpose distributes over both combination rules, both spaces are subspaces of the symmetric matrices. Proving that it belongs to $\SYR{2nm}$ follows by using that $J_n \otimes I_m =  I_n \otimes J_m $ in conjunction with \Cref{lem:2mix_absorbant} to show the commutation with $J_{nm}$ in accordance with \Cref{lem:SYM_char}.
		Since both combinations behave similarly as the tensor product of complex matrices under the trace, \cref{eq:rwedgetensor_trace}, and since each factor is locally isometrically isomorphic to the space of Hermitian matrices (\cref{prop:isometry}), an $n^2 \times m^2$ basis can be formed by composing a basis vectors of $\SYR{n}$ with one of $\SYR{m}$ under $\rtensor$. Thus, we have obtained a set of $n^2m^2$ linearly independent and pairwise orthogonal vectors of $\SYR{2nm}$ that can be used to span a subspace isomorphic to $\SYR{nm}$.
	\end{proof}
	
	Therefore, there exist ways to linearly map the combination of $\Gi{n}{\HermC{n}}$ and $\Gi{m}{\HermC{m}}$ into $\Gi{nm}{\HermC{nm}}$. There is however essentially one way that is ``better'' than the others: the one preserving the combination rule through $\Gs$. Define $\tau_{\rtensor} : \rr^{4nm \times 4nm} \rightarrow \rr^{2nm \times 2nm}: \tau_{\rtensor}(\SYR{n} \rtensor \SYR{m}) = \SYR{nm}$ as this map. That is, such that $\forall H\in \HermC{n}, \forall L \in \HermC{m}$:
	\begin{equation}
		\Gi{nm}{H \otimes L} = \tau_{\rtensor}(\Gi{n}{H} \rtensor \Gi{m}{L}) \:.
	\end{equation}
	
	Looking at the characterization \eqref{eq:2mix_Re_Im} of the $\rwedgetensor$ combination, there is an obvious choice for this map. Let $\I{2}_d := \I{2} \otimes \id_d$ and $\J{2}_d := \J{2} \otimes \id_d$, then
	\begin{equation}
		\tau_{\rwedgetensor}(\cdot) := \id \otimes \frac{\TrX{\rr^{2 \times 2}\otimes \rr^{2 \times 2}}{\I{2}_{nm} \: \cdot}}{2} + J \otimes \frac{\TrX{\rr^{2 \times 2}\otimes \rr^{2 \times 2}}{\J{2}_{nm} \: \cdot}}{2} \:,
	\end{equation}
	where the subscript notation in the traces indicates a partial trace over only the two $\rr^{2\times 2}$ systems in which the representation of $\{1,i\}$ is defined. This way, its action on a $\rwedgetensor$ combination gives
	\begin{equation}\label{eq:dot_char}
		\tau_{\rwedgetensor}(\Gi{n}{H} \rwedgetensor \Gi{m}{L}) = \id \otimes \ReP{H \otimes L} + J \otimes \ImP{H\otimes L}\:,
	\end{equation}
	which is none other than $\Gi{nm}{L \otimes H}$. 
	We have defined a new combination rule out of this procedure, and, by construction, it will be free from the issues encountered with the tensor product.
	
	\begin{prop}\label{prop:propertiesofdottensor}
		For $\widetilde{H} \in \SYR{n}$ and $\widetilde{L} \in \SYR{m}$, let the map $\rdottensor: \rr^{2n \times 2n} \times \rr^{2m \times 2m} \rightarrow \rr^{2nm}$ be defined by
		\begin{equation}
			\widetilde{H} \rdottensor \widetilde{L} := \tau_{\rwedgetensor}(\widetilde{H} \rwedgetensor \widetilde{L}) \:.
		\end{equation}
		This map is the combination law for the special symmetric matrices in the following senses:
		\begin{enumerate}
			\item it is real-linear in both inputs, associative, and obeys the mixed product rule;
			\item it is compatible with the inner product (up to a rescaling factor), $\TrX{}{\widetilde{H} \rdottensor \widetilde{L}} = 2 \TrX{}{\widetilde{H}}\TrX{}{\widetilde{L}}$;
			\item $\SYR{n} \rdottensor \SYR{m} = \SYR{nm}$ and in particular $\dim{\SYR{n} \rdottensor \SYR{m}} = \dim{\SYR{n}}\dim{\SYR{m}}$.
		\end{enumerate}
	\end{prop}
	\begin{proof}
		Using the characterization \eqref{eq:dot_char}, items 1. and 2. directly follow from the properties of the tensor product, except for the associativity, which can be proven in the following way: Let $M_1,M_2M_3$ be arbitrary Hermitian matrices, then
		\begin{multline}
			(\Gi{d_1}{M_1}\rdottensor \Gi{d_2}{M_2}) \rdottensor\Gi{d_3}{M_3} =\\ \left(    \id \otimes\ReP{ M_1 \otimes M_2} + J\otimes \ImP{  M_1\otimes M_2}  \right) \rdottensor \left(    \id \otimes\ReP{ M_3} + J\otimes \ImP{  M_3}  \right) =\\
			\id \otimes \left(  \ReP{M_1 \otimes M_2 } \otimes \ReP{ M_3} - \ImP{ M_1\otimes M_2}\otimes\ImP{ M_3}  \right) +\\
			J \otimes \left( \ReP{ M_1\otimes M_2} \otimes \ImP{ M_3} - \ImP{  M_1\otimes M_2} \otimes \ReP{ M_3}  \right) =\\
			\id \otimes \ReP{ (M_1\otimes M_2) \otimes M_3} + J\otimes \ImP{ (M_1\otimes M_2)\otimes M_3} \:.\label{eq: assoco proof line:last}
		\end{multline}
		Similarly, we find
		\begin{align}
			\Gi{d_1}{M_1}\rdottensor (\Gi{d_2}{M_2} \rdottensor\Gi{d_3}{M_3}) = \id \otimes \ReP{ M_1\otimes( M_2 \otimes M_3)} + J\otimes \ImP{ M_1\otimes (M_2\otimes M_3)}.\label{eq: assoco proof line:last2}
		\end{align}
		Therefore, by comparing equations~\labelcref{eq: assoco proof line:last,eq: assoco proof line:last2}, and taking into account the associativity of the $\otimes$ product, associativity holds.
		
		Item 3. follows from \cref{prop:iso_spaces} and realizing that the map $\tau_{\rwedgetensor}$ is an isomorphism (between two representations of $\{1,i\}$).
	\end{proof}
	
	\section{Explicit construction of the \texorpdfstring{$M(\cdot)$}{M(•)} map}\label{sec:mapping}
	Using that the real vector spaces $\Herm_{d_l}(\cc)$ and $\SY_{d_l}(\rr)$ are isomorphic, the linear identification between the two can be recast as a `base change' through the Hilbert-Schmidt inner product:
	\begin{equation}
		\begin{aligned}
		&\Gs_{d_l} : \Herm_{d_l}(\cc) \to \SY_{d_l}(\rr)\:,\\
		&\Gs_{d_l}(\cdot):=\sqrt{2} \sum_{j=1}^{{d_l}^2} \Tr[H_j^\dag \: \cdot\: ] S_j \:,\label{eq:map local}
		\end{aligned}
	\end{equation}
	where $\{H_j\}_{j=1}^{d_l^2}$ is an orthonormal basis for $\Herm_{d_l}(\cc)$, namely $\Tr[H_j^\dag H_k]=\delta_{j,k}$ for all $j,k\in\{1,\ldots,d_l\}$, and similarly $\{S_j\}_{j=1}^{d_l^2}$ is an orthonormal basis for $\SY_{d_l}(\rr)$, namely $\Tr[S_j^T S_k]=\delta_{j,k}$ for all $j,k\in\{1,\ldots,{d_l}\}$.
	It follows from the orthogonality relations that said map is invertible, with inverse 
	\begin{equation}
		\begin{aligned}
			&\Gs_{d_l}^{-1} : \SY_{d_l}(\rr)\to \Herm_{d_l}(\cc) \:,\\
			&\Gs_{d_l}^{-1}(\cdot):=\frac{1}{\sqrt{2}}\sum_{j=1}^{d_l^2} \Tr[S_j^T \: \cdot\: ] H_j \:.\label{eq:map local inverse}
		\end{aligned}
	\end{equation}
	
	We now define the sets $\{H_j\}_{j=1}^{{d_l}^2}$  and $\{S_j\}_{j=1}^{d_l^2}$ explicitly so that the map~\cref{eq:map local} realizes the map $\Gamma_{d_l}$ of the main text, namely~\cref{eq:local mapping}, up to a constant factor. With this task in mind, we re-write the map using double-index notation as follows:
	\begin{equation}
		\Gs_{d_l}(\cdot)=\sqrt{2} \left( \sum_{j=1}^{d_l} \Tr[(\cdot) H_{jj}^\textup{Di} ] S_{jj}^\textup{Di} \:+\sum_{1\leq j < k \leq d_l} \Tr[(\cdot) H_{jk}^\textup{R} ] S_{jk}^\textup{R} \:+ \sum_{1\leq j < k \leq d_l} \Tr[(\cdot) H_{jk}^\textup{I} ] S_{jk}^\textup{I}\right) \:, \label{eq:pos to pos channel}
	\end{equation}
	where 
	\begin{subequations}
	\begin{align}
		H_{jj}^\textup{Di} &= \proj{j}\:,\\
		S_{jj}^\textup{Di} &=\frac{\id}{\sqrt{2}} \otimes \dyad{j} = \frac{1}{\sqrt{2}} \big( \proj{j}_{11}+ \proj{j}_{22} \big)\:, \\
		H_{jk}^\textup{R} &= \frac{1}{\sqrt{2}}\big( \ketbra{j}{k}+ \ketbra{k}{j} \big)\:,\\
		S_{jk}^\textup{R} &= \frac{\id}{\sqrt{2}} \otimes \frac{1}{\sqrt{2}}\big( \ketbra{j}{k}+ \ketbra{k}{j} \big) = \frac{1}{2} \big( \ketbra{j}{k}_{11}+ \ketbra{k}{j}_{11}+ \ketbra{j}{k}_{22}+ \ketbra{k}{j}_{22} \big)\:, \\
		H_{jk}^\textup{I} &= \frac{i}{\sqrt{2}}\big( \ketbra{j}{k}- \ketbra{k}{j} \big),\\
		S_{jk}^\textup{I} &= \frac{J}{\sqrt{2}} \otimes \frac{1}{\sqrt{2}}\big( \ketbra{j}{k}- \ketbra{k}{j} \big) = \frac{1}{2} \Big( -\big(\ketbra{j}{k}_{12}- \ketbra{k}{j}_{12} \big)+ \big(\ketbra{j}{k}_{21} - \ketbra{k}{j}_{21} \big) \Big),	
	\end{align}
	\end{subequations}
	and where we are using block matrix notation for the symmetric matrices. I.e. a generic matrix in $\SY_{d_l}(\rr)$ takes on the form
	\begin{align}
		\left({\begin{array}{c|c} A_{11} & A_{12} \\\hline A_{21} & A_{22} \end{array}}\right),
	\end{align}
	where $A_{lm}= \sum_{j,k=1}^{d_l} a_{jk} \ketbra{j}{k}_{lm}$, $a_{jk}\in\rr$. It is straightforward to verify that the above equations are indeed of the form \cref{eq:local mapping} up to a normalisation factor of $\sqrt{2}$. First note that in \cref{eq:pos to pos channel}, we use $d_l$ diagonal matrices of the form $H_{jj}^\textup{Di}$, $(d_l^2-d_l)/2$ symmetric matrices of the form $H_{jk}^\textup{R}$, and $(d_l^2-d_l)/2$ antisymetric matrices of the form $H_{jk}^\textup{I}$. Thus a total of $d_l^2$, the vector sapce dimension $\Herm_{d_l}(\cc)$. Furthermore, by direct calculation we see that they are all Hermitian and orthonormal under the Hilbert-Schmidt inner product $\big(\!<\cdot, \cdot> =\Tr[ (\cdot)^\dag (\cdot)]\,\big)$. Therefore, they form a complete orthonormal basis for $\Herm_{d_l}(\cc)$. Similarly, using the relationships
	\begin{align}
		&\big(\ketbra{j}{k}_{lm} \big)^\dag= \ketbra{k}{j}_{ml}, \quad \Tr\big[\ketbra{j}{k}_{lm} \big]=\delta_{jk}\delta_{lm},\\
		&\ketbra{j}{k}_{lm} \ketbra{q}{r}_{np}= \ketbra{j}{r}_{lp} \delta_{kq} \delta_{mn},
	\end{align}
	we can straightforwardly verify that the matrices $S_{jj}^\textup{Di}$, $S_{jk}^\textup{R}$, $S_{jk}^\textup{I}$ appearing in \cref{eq:pos to pos channel} are real, symmetric and orthonormal under the Hilbert-Schmidt inner product. They therefore span a $d_l^2$-dimensional  subspace of the symmetric matrices of dimension $2d_l$, namely $\SY_{d_l}(\rr)$. This verifies that \cref{eq:pos to pos channel} realizes map~\cref{eq:local mapping}.
	
	Applying the local mapping $\Gs_{d_l}$ on $n$ subsystem is then realized via \cref{eq:map M}. 
	The inverse of $M$, $M^{-1}(\cdot) \!: \SY_{d_l}(\rr)\!\to \Herm_{d_l}(\cc)$ is thus
	\begin{align}
		M^{-1}(\cdot)=  \Gamma_{d_1}^{-1}(\cdot) \otimes \Gamma_{d_2}^{-1}(\cdot)\otimes \ldots\otimes  \Gamma_{d_n}^{-1}(\cdot)/a.
	\end{align}
	This can be easily verified by noting that $M^{-1}\circ M(\cdot)$ and $M\circ M^{-1}(\cdot)$ are the identity maps for their respective spaces. 
	Specifically, we have 
	\begin{align}
		M^{-1}\circ M(\cdot)= (\Gamma_{d_1}^{-1}\circ \Gamma_{d_1})(\cdot) \otimes (\Gamma_{d_2}^{-1}\circ \Gamma_{d_2})(\cdot)\otimes \ldots\otimes (\Gamma_{d_n}^{-1}\circ \Gamma_{d_n})(\cdot)= \id_{\textup{C},d_1}\otimes \id_{\textup{C},d_2}\otimes \ldots \otimes \id_{\textup{C},d_n},
	\end{align}
	where $\id_{\textup{C},d_l}$ is the identity operator on  $\Herm_{d_l}(\cc)$. Similarly,  
	\begin{align}
		M\circ M^{-1}(\cdot)= (\Gamma_{d_1}\circ \Gamma_{d_1}^{-1})(\cdot) \rtensor (\Gamma_{d_2}\circ \Gamma_{d_2}^{-1})(\cdot)\rtensor \ldots\rtensor (\Gamma_{d_n}\circ \Gamma_{d_n}^{-1})(\cdot)= \id_{\textup{R},d_1}\rtensor \id_{\textup{R},d_2}\rtensor \ldots \rtensor \id_{\textup{R},d_n},\label{eq:identity on mapped space}
	\end{align}
	where $\id_{\textup{R},d_l}$ is the identity operator on $\SY_{d_l}(\rr)$. The r.h.s. of~\cref{eq:identity on mapped space} is the identity operator on $\SY_{n,d}(\rr)$.

\section{Properties of real separable states}\label{app:properties of separable states}
Here we prove two lemmas relating to the structure of real theories.
\begin{lemma}\label{lem:separable states adn effects}
	The sets of separable states and effects in real theories and CQT are in one-to-one correspondence, namely $\rho\in\Ssep$ iff $M^{-1}(\rho)\in\SsepC$, and $e\in\Esep$ iff $M^{-1}(e)\in\EsepC$.
\end{lemma}
\begin{proof}
	We will only write the proof for effects since it is analogous for states. The proof for effects will consist in showing $e\in\Esep$ $\implies$ $M^{-1}(e)\in\EsepC$  followed by proving the converse relationship.
	$\rho\in\Ssep$ implies by definition that $\rho$ can be written in the form $e= \sum_j p_j \,e_j^1\rtensor e_j^2\rtensor\ldots\rtensor e_j^n$. Therefore, using the linearity of $M^{-1}(\cdot)$, it implies $M^{-1}(e)= \sum_j p_j M^{-1}(e_j^1\rtensor e_j^2\rtensor\ldots\rtensor e_j^n)$ $=\sum_j p_j \Gamma^{-1}_{d_1}(e_j^1)\otimes \Gamma^{-1}_{d_2}(e_j^2)\otimes\ldots\otimes \Gamma^{-1}_{d_n}(e_j^n)$. The quantity on the r.h.s. of the last equality is, by definition, in $\EsepC$, thus proving $e\in\Esep$ $\implies$ $M^{-1}(e)\in\EsepC$. Similarly, it follows that $M^{-1}(e)\in\EsepC$ $\implies$ $e\in\Esep$.
\end{proof}

For the next lemma, it is useful to introduce some notation: let $\mathsf{SepStatRe}:=\{  \Tr[e \rho] : e\in\EffR, \rho\in\Ssep \}$ denote the set of statistics generatable in an RRQT by the set of separable states in said theory. Likewise, let $\mathsf{SepStatCom}:=\{  \Tr[e \rho] : e\in\EffCom, \rho\in\SsepC \}$ denote the set of statistics generatable in a CQT by the set of separable states in said theory.

Completely analogously, we define the same quantities for the statistics generated by separable measurements in an RRQT and in CQT: $\mathsf{SepEffRe}:=\{  \Tr[e \rho] : e\in\Esep, \rho\in\StatR \}$ and $\mathsf{SepEffCom}:=\{  \Tr[e \rho] : e\in \EsepC, \rho\in\StatCom \}$ respectively.

\begin{lemma}
	If a real theory is a RNQT, (namely if~\cref{eq:inner product} holds), then $\mathsf{SepStatRe}=\mathsf{SepStatCom}$ and $\mathsf{SepEffRe}=\mathsf{SepEffCom}$. This is to say, the statistics of its separable states coincide with that of CQT. 
\end{lemma}
\begin{proof}
	We only prove it for the case  $\mathsf{SepEffRe}=\mathsf{SepEffCom}$ since the proof is analogous for $\mathsf{SepStatRe}=\mathsf{SepStatCom}$. 
	
	Consider an effect $e\in\Esep$, and state $\rho\in\StatR$. Since the theory is a RNQT, we have by definition that
	\begin{align}
		\Tr[e \rho]=\Tr[E^{-1}(e) M^{-1}(\rho)]=\Tr[M^{-1}(e) M^{-1}(\rho)]/a.\label{eq:inner equivalence}
	\end{align}
	However, we have from~\Cref{lem:separable states adn effects}, that $e\in\Esep$ implies $M^{-1}(e)\in\EsepC$. Thus~\cref{eq:inner equivalence} implies that all elements in $\mathsf{SepEffRe}$ are also in $\mathsf{SepEffCom}$.  Similarly, one can show that all elements in $\mathsf{SepEffCom}$ are also in $\mathsf{SepEffRe}$. We thus conclude that $\mathsf{SepEffRe}=\mathsf{SepEffCom}$.
\end{proof}
	
\section{Proof of~\texorpdfstring{\Cref{thm:no need for complex numbers}}{Theorem 2}}\label{Sec:proof of state set determination without complext numbers}
	\begin{proof}
		Since we will be dealing with the map $M$ in~\cref{eq:map M} for different instantiations of the combination rule $\rtensor$ in this proof, we add a subscript to distinguish between the different instantiations. In the generic case we denote $M(\cdot)$ by $M_{\rtensor}(\cdot)$, while in the specialized case $\rtensor=\rdottensor$, we denote it by $M_{\rdottensor}(\cdot)$. We also add these subscripts to the other quantities which depend on the instantiation of $\rtensor$.s
		
		Recall that, by definition, a vector $v_r\in\SY_{n,d,{\rtensor}}(\rr)$ is in ${\StatR}_{\rtensor}$ iff
		\begin{align}
			M_{\rtensor}^{-1}(v_r):=v_c\geq 0,\label{eq: vc def}
		\end{align}
		where $v_c\in\Herm_{n,d}(\cc)$ and $\Herm_{n,d}(\cc)$ is a vector space which is independent of $\rtensor$.  Rearranging, we have $v_r=M_{\rtensor}(v_c)$.
		Now let us apply $G(\cdot)$ to both sides and use the associativity properties of $G(\cdot)$ and $M(\cdot)$ (which follow from the associativity of $\rtensor$). We find
		\begin{align}
			G(v_r)=G(M_{\rtensor}(v_c))= M_{\rdottensor}(v_c).\label{eq:G applies to map}
		\end{align}
		Let us now repeat the first lines of reasoning but now for the special case   $\rtensor=\rdottensor$ and for the same vector $v_c\in\Herm_{n,d}$:  a vector $v_r'\in\SY_{n,d,{\rdottensor}}(\rr)$ is in ${\StatR}_{\rdottensor}$ iff $M_{\rdottensor}^{-1}(v_r'):=v_c\geq 0$. However, we also have from~\Cref{thm:new combination rule} that $v_r'\in{\StatR}_{\rdottensor}$ iff $v_r'\geq 0$. Thus $v_r'\geq 0$ iff $v_c\geq 0$, where $v_r'=M_{\rdottensor}(v_c)$. Or in other words, $M_{\rdottensor}(v_c)\geq 0$ iff $v_c\geq 0$.
		Thus, from~\cref{eq: vc def,eq:G applies to map}, we find
		\begin{align}
			G(v_r)\geq 0   \quad\text{iff}\quad M_{\rtensor}^{-1}(v_r) \geq 0.
		\end{align}
		
	\end{proof}

	\section{Proof of~\texorpdfstring{\Cref{thm:normal tensor product}}{Theorem 3}}\label{sec:proof thm normal tensor product}
	\begin{proof}
		We first start by showing that the choice $\rtensor=\otimes$ gives rise to a RNQT. For this,
		we need to show that~\cref{eq:inner product} holds. Due to the linearity of $M^{-1}$ and the trace, it suffices to prove the equality holds for all product states and effects. Using~\cref{eq:map local inverse} we find
		\begin{align}
			&\Tr[M^{-1}(e_1\otimes e_2\otimes \ldots \otimes e_n) M^{-1}(\rho_1\otimes\rho_2\otimes \ldots \otimes \rho_n)]
			\\&= \Tr[\Gamma_{d_1}^{-1}(e_1) \Gamma_{d_1}^{-1}(\rho_1) ]   \Tr[\Gamma_{d_2}^{-1}(e_2) \Gamma_{d_2}^{-1}(\rho_2) ] \ldots  \Tr[\Gamma_{d_n}^{-1}(e_n) \Gamma_{d_n}^{-1}(\rho_n) ].\label{eq:fisrt trace eq}
		\end{align}
		Let us now focus on the $l$th term individually:
		\begin{align}
			\Tr[\Gamma_{d_l}^{-1}(e_l) \Gamma_{d_l}^{-1}(\rho_l) ]=\sum_{j,k=1}^{d_l^2}  \Tr[e_l S_j]  \Tr[\rho_l S_k] \Tr[ H_j   H_k]=\sum_{j=1}^{d_l^2}  \Tr[e_l S_j]  \Tr[\rho_l S_j]= \Tr[e_l\rho_l],  
		\end{align}
		where in the last line we have used the resolution of the identity, namely
		\begin{align}
			\rho_l=\sum_{j=1}^{d_l^2} \Tr[\rho_l S_j] S_j.
		\end{align}
		Thus from~\cref{eq:fisrt trace eq}, we find
		\begin{align}
			&\Tr[M^{-1}(e_1\otimes e_2\otimes \ldots \otimes e_n) M^{-1}(\rho_1\otimes\rho_2\otimes \ldots \otimes \rho_n)]=  \Tr[e_1\rho_1] \Tr[e_2\rho_2]\ldots  \Tr[e_n\rho_n]
			\\&=\Tr[(e_1\otimes e_2\otimes \ldots \otimes e_n) (\rho_1\otimes\rho_2\otimes \ldots \otimes \rho_n)],\label{eq:fisrt trace eq2}
		\end{align}
		thus verifying~\cref{eq:inner product} for all states in $\StatR$ and effects in $\EffR$ of the product-state form.
		
		We now move on to prove that if a state is separable, then it is positive semi-definite, namely if $\rho\in\Ssep$, then $\rho \geq 0$.  For this, we recall the definition of states in $\Ssep$, namely from~\Cref{sec:Structure} we have that $\rho\in\Ssep$ is of the form
		\begin{align}
			\rho= \sum_j p_j \,\rho_j^1\rtensor \rho_j^2\rtensor\ldots\rtensor \rho_j^n= \sum_j p_j \,\rho_j^1\otimes \rho_j^2\otimes\ldots\otimes \rho_j^n,\label{eq:tensor product of RQT states}
		\end{align} 
		where in the last line, we have used the fact that $\rtensor=\otimes$ in this RNQT by definition, and where for the $l$th term $\Gamma^{-1}_{d_l}(\rho_j^l) \geq 0$, since $\Gamma^{-1}_{d_l}(\rho_j^l)$ is in $\StatCom$. Therefore, denoting $\Gamma^{-1}_{d_l}(\rho_j^l) =: \rho_\textup{Re}+i \rho_\textup{Im}$, it follows  (due to~\cref{eq:local mapping}) that
		\begin{align}
			\rho_j^l = \Gamma_{d_l}(\rho_\textup{Re}+i \rho_\textup{Im})= 
			\left({\begin{array}{c|c} \rho_\textup{Re} & -\rho_\textup{Im} \\\hline \rho_\textup{Im} & \rho_\textup{Re} \end{array}}\right),
		\end{align}
		where $\rho_\textup{Re}+i \rho_\textup{Im} \geq 0$.
		
		We have seen in~\Cref{thm:posvitive semi definitness conservation of the local mapping} that  $\rho_\textup{Re}+i \rho_\textup{Im} \geq 0$  implies $\Gamma_{d_l}(\rho_\textup{Re}+i \rho_\textup{Im}) \geq 0$.    Therefore, since the coefficients $p_j$ in~\cref{eq:tensor product of RQT states} are non negative, and the tensor product of positive semi-definite matrices is itself positive semi-definite, we conclude that $\rho$ in~\cref{eq:tensor product of RQT states} is positive semi-definite. Thus we have proven that  if $\rho\in\Ssep$, then $\rho \geq 0$. 
		
		Similarly, one can show that  if $e\in\Esep$, then $e \geq 0$.

		Finally, all that is left to prove is that there exist entangled effects and entangled states in $\EffR$ and $\StatR$ respectively, which are not positive semi-definite.  We do this by explicit construction. Since the set $\StatR$ is a partition of the set of real separable states and real entangled states, and we have proven that all separable states are positive semi-definite, it suffices to prove that there exists a state in $\StatR$ which has at least one negative eigenvalue. This is what we will now do. The proof follows analogously for the existence of an entangled effect which has at least one negative eigenvalue.
		
		Consider the Bell state for two qubits, in CQT:  $\ket{\Phi^{+}}= \frac{1}{\sqrt{2}} (\ket{00}+\ket{11})$. Or in density matrix form
		\[
		|\Phi^+\rangle\langle\Phi^+| = \frac{1}{2}
		\begin{pmatrix}
			1 & 0 & 0 & 1 \\[4pt]
			0 & 0 & 0 & 0 \\[4pt]
			0 & 0 & 0 & 0 \\[4pt]
			1 & 0 & 0 & 1
		\end{pmatrix}.
		\]
		We can compute exactly the corresponding density matrix in $\SY_{n,d}(\rr)$ for this case together with its eigenvalues. We find
		\begin{align}
			\Gamma_{2}\otimes \Gamma_{2} \,(|\Phi^+\rangle\langle\Phi^+|)= 
			\left(
			\begin{array}{cccccccccccccccc}
				\frac{1}{2} & 0 & 0 & 0 & 0 & \frac{1}{4} & 0 & 0 & 0 & 0 & 0 & 0 & 0 & 0 & 0 & -\frac{1}{4} \\
				0 & 0 & 0 & 0 & \frac{1}{4} & 0 & 0 & 0 & 0 & 0 & 0 & 0 & 0 & 0 & \frac{1}{4} & 0 \\
				0 & 0 & \frac{1}{2} & 0 & 0 & 0 & 0 & \frac{1}{4} & 0 & 0 & 0 & 0 & 0 & \frac{1}{4} & 0 & 0 \\
				0 & 0 & 0 & 0 & 0 & 0 & \frac{1}{4} & 0 & 0 & 0 & 0 & 0 & -\frac{1}{4} & 0 & 0 & 0 \\
				0 & \frac{1}{4} & 0 & 0 & 0 & 0 & 0 & 0 & 0 & 0 & 0 & \frac{1}{4} & 0 & 0 & 0 & 0 \\
				\frac{1}{4} & 0 & 0 & 0 & 0 & \frac{1}{2} & 0 & 0 & 0 & 0 & -\frac{1}{4} & 0 & 0 & 0 & 0 & 0 \\
				0 & 0 & 0 & \frac{1}{4} & 0 & 0 & 0 & 0 & 0 & -\frac{1}{4} & 0 & 0 & 0 & 0 & 0 & 0 \\
				0 & 0 & \frac{1}{4} & 0 & 0 & 0 & 0 & \frac{1}{2} & \frac{1}{4} & 0 & 0 & 0 & 0 & 0 & 0 & 0 \\
				0 & 0 & 0 & 0 & 0 & 0 & 0 & \frac{1}{4} & \frac{1}{2} & 0 & 0 & 0 & 0 & \frac{1}{4} & 0 & 0 \\
				0 & 0 & 0 & 0 & 0 & 0 & -\frac{1}{4} & 0 & 0 & 0 & 0 & 0 & \frac{1}{4} & 0 & 0 & 0 \\
				0 & 0 & 0 & 0 & 0 & -\frac{1}{4} & 0 & 0 & 0 & 0 & \frac{1}{2} & 0 & 0 & 0 & 0 & \frac{1}{4} \\
				0 & 0 & 0 & 0 & \frac{1}{4} & 0 & 0 & 0 & 0 & 0 & 0 & 0 & 0 & 0 & \frac{1}{4} & 0 \\
				0 & 0 & 0 & -\frac{1}{4} & 0 & 0 & 0 & 0 & 0 & \frac{1}{4} & 0 & 0 & 0 & 0 & 0 & 0 \\
				0 & 0 & \frac{1}{4} & 0 & 0 & 0 & 0 & 0 & \frac{1}{4} & 0 & 0 & 0 & 0 & \frac{1}{2} & 0 & 0 \\
				0 & \frac{1}{4} & 0 & 0 & 0 & 0 & 0 & 0 & 0 & 0 & 0 & \frac{1}{4} & 0 & 0 & 0 & 0 \\
				-\frac{1}{4} & 0 & 0 & 0 & 0 & 0 & 0 & 0 & 0 & 0 & \frac{1}{4} & 0 & 0 & 0 & 0 & \frac{1}{2} \\
			\end{array}
			\right),
		\end{align} 
		with eigenvalues 
		\begin{align}
			\left\{1,1,\frac{1}{2},\frac{1}{2},\frac{1}{2},\frac{1}{2},\frac{1}{2},\frac{1}{2},0,0,0,0,0,0,-\frac{1}{2},-\frac{1}{2}\right\}.
		\end{align}
		Thus the existence of the negative eigenvalue of $-1/2$ concludes the proof.
	\end{proof}

\end{document}